\documentclass[11pt]{article}
\usepackage{amsmath}
\usepackage{amsmath}
\usepackage{times}
\usepackage{subfigure}
\usepackage{fullpage}
\usepackage{color}
\usepackage{graphicx,amssymb,amsmath}
\usepackage{multirow}
\usepackage{xspace}
\usepackage[linesnumbered,boxed,ruled, vlined]{algorithm2e} 
\usepackage[top=1in, bottom=1in, left=1in, right=1in]{geometry}
\usepackage{cases}

\newcommand{\commentout}[1]{}
\newcommand{\eat}[1]{}
\newcommand{\topic}[1]{\noindent{{\bf #1}:}}

\newcommand{\calH}{{\mathcal H}}

\newcommand{\calF}{{\mathcal F}}
\newcommand{\calP}{{\mathcal P}}

\newcommand{\calT}{{\mathcal T}}

\newcommand{\calR}{{\mathcal R}}

\newcommand{\Prob}{{\operatorname{Pr}}}
\newcommand{\Exp}{{\mathbb{E}}}
\newcommand{\E}{{\mathbb{E}}}

\renewcommand{\H}{{\mathcal{H}}}

\renewcommand{\P}{{\mathcal P}}

\newcommand{\hpf}{\Psi(\calP)}

\newcommand{\dist}{\mathrm{d}}
\renewcommand{\d}{\mathrm{d}}

\newcommand{\rank}{\mathsf{rank}}

\newcommand{\FA}{\mathrm{FPRAS}}
\newcommand{\sharpP}{\#\mathrm{P}}
\newcommand{\inapp}{\mathrm{Inapprox}}
\newcommand{\open}{\mathrm{Open}}

\newcommand{\B}{\mathsf{B}}

\newcommand{\A}[1]{\langle #1\rangle}
\newcommand{\e}{\epsilon}
\newcommand{\bS}{\bar{S}}

\newcommand{\V}{\mathcal{V}}
\renewcommand{\L}{\mathsf{T}}

\newcommand{\iin}{\mathsf{In}}
\newcommand{\poly}{\mathrm{poly}}

\newcommand{\diam}{\mathrm{diam}}
\newcommand{\D}{\mathsf{D}}
\newcommand{\C}{\mathsf{C}}

\newcommand{\KCP}{\mathsf{kC}}
\newcommand{\CP}{\mathsf{C}}
\newcommand{\KC}{\mathsf{kCL}}
\newcommand{\MST}{\mathsf{MST}}
\newcommand{\MM}{\mathsf{PM}}
\newcommand{\CC}{\mathsf{CC}}

\newcommand{\NN}{\mathsf{NN}}
\newcommand{\KMNN}{\mathsf{kmNN}}

\newcommand{\core}{stoch-core}

\renewcommand{\d}{\mathrm{d}}
\newcommand{\p}{p}
\newcommand{\Cl}{\mathsf{Cl}}
\newcommand{\Home}{\mathcal{H}}

\newcommand{\realize}{\vDash}
\newcommand{\consistent}{\thicksim}

\renewcommand{\r}{{\mathbf{r}}}

\newcommand{\per}{\mathrm{Per}}

\newenvironment{proof}{\noindent {\em Proof: }\ignorespaces}{}

\newcommand{\qed}{\hspace*{\fill}$\Box$\medskip}

\newtheorem{theorem}{Theorem}

\newtheorem{lemma}{Lemma}

\newtheorem{definition}{Definition}

\title{Approximating the Expected Values for Combinatorial Optimization Problems over Stochastic Points}

\author{Lingxiao Huang \quad\quad\quad\quad Jian Li \\
Institute for Interdisciplinary Information Sciences\\
Tsinghua University, China
}

\begin{document}

\pagenumbering{gobble}
\begin{titlepage}

\maketitle

\begin{abstract}
\normalsize{
We consider the stochastic geometry model where the location of each node
is a random point in a given metric space, or the existence of each node is uncertain.
We study the problems of computing
the expected lengths of several combinatorial or geometric optimization problems
over stochastic points, including closest pair,
minimum spanning tree, $k$-clustering, minimum perfect matching, and minimum cycle cover.
We also consider the problem of
estimating the probability that
the length of closest pair, or the diameter,
is at most, or at least, a given threshold.
Most of the above problems are known to be $\sharpP$-hard.
We obtain FPRAS (Fully Polynomial Randomized
Approximation Scheme) for most of them
in both the existential and locational uncertainty models.
Our result for stochastic minimum spanning trees in the locational uncertain model
improves upon the previously known constant factor approximation
algorithm. Our results for other problems
are the first known to the best of our knowledge.
}

\eat{
We mainly use two new techniques.
\begin{enumerate}

\item For estimating $\Prob[\CP\leq 1]$, the stochastic minimum spanning tree, and the minimum perfect matching problems, we propose a technique called `` stochastic \core". We first find out the ``\core" from the point set, such that all present points are in the ``\core" with high probability. Then we decompose the problem into a convex combination of conditional expectations using this ``\core", and show that we only need to consider the situation that only very few present points are outside the all present points are in the ``\core".
\item For stochastic closest pair and the $k$-clustering, we propose another technique called ``Hierarchical decomposition tree (HDT)". We first construct a minimum spanning tree on the point set. We build an HDT structure with $m$ levels based on the minimum spanning tree. We partition all instances of the stochastic problem into $m$ parts based on different levels of the HDT structure. Then we show that for each part of instance, we can take enough samples to achieve a good estimation. Combine all these estimation, we can obtain an FPRAS.
\end{enumerate}
}
\end{abstract}
\end{titlepage}

\newpage

\pagenumbering{arabic}
\setcounter{page}{1}

\vspace{-0.3cm}
\section{Introduction}

\topic{Background}
Uncertain or imprecise data are pervasive in applications like sensor monitoring, location based services, data collection and integration~\cite{cheng2008cleaning,dong2007data,suciu2011probabilistic}.
Consider a sensor network deployed in the wild to monitor the living habits or migration of certain animals~\cite{ mainwaring2002wireless,szewczyk2004habitat}.
Since sensing instruments are not perfect, the data collected are often contaminated with a significant amount of noise~\cite{deshpande2004model,szewczyk2004habitat}.
For another example, the locational data collected by the Global-Positioning Systems (GPS) often contains measurement errors~\cite{pfoser1999capturing}.
Moreover, many machine learning and prediction algorithms also produce a variety of stochastic models and a large volume of probabilistic data.
Thus, managing, analyzing and solving optimization problems over stochastic models and data have recently attracted significant attentions in several research communities
(see e.g., \cite{shapiro2014lectures,suciu2011probabilistic, swamy2006approximation}).

In this paper, we study two stochastic geometry models,
the locational uncertainty model and the existential uncertainty model,
both of which have been studied extensively
in recent years (see e.g.,
\cite{agarwal2014convex,agarwal2012range,agarwal2009indexing, atallah2011asymptotically, kamousi2011stochastic, kamousi2014closest, li2010ranking, li2014epsilon, li2014range},
some of which will be discussed in the related work section).
In fact, a special case of the locational uncertainty model where all points follow the same distribution is a classic topic in stochastic geometry literature (see e.g., \cite{beardwood1959shortest, bern1993worst-case,bertsimas1990asymptotic, karloff1989how,snyder1995aprioribounds}).
The main interest there has been to derive asymptotics for the expected values of certain
combinatorial problems (e.g., minimum spanning tree).
The stochastic geometry model is also of fundamental interest in the area of wireless networks. In many applications, we only
have some prior information about the locations of the transmission nodes (e.g., some sensors that will
be deployed randomly in a designated area by an aircraft). Such a stochastic wireless network can be
captured precisely by this model. See the recent survey \cite{haenggi2009stochastic} and more references therein.

\vspace{0.1cm}
\topic{Stochastic Geometry Models}
In this paper, we focus on two stochastic geometry models, the locational uncertainty model and existential uncertainty model.
\begin{enumerate}
\item (Locational Uncertainty Model)
We are given a metric space $\calP$.
The location of each node $v\in\V$
is a random point in the metric space $\calP$ and the probability distribution
is given as the input.
Formally, we use the term {\em nodes} to refer to the vertices of the graph,
{\em points} to describe the locations of the nodes in the metric space. We
denote the set of nodes as $\V=\{v_1,\ldots, v_n\}$
and the set of points as $\P=\{s_1, \ldots, s_m\}$,
where $n=|\V|$ and $m=|\P|$.
A realization $\r$  can be represented by an $n$-dimensional
vector $(r_1, \ldots, r_n) \in \P^n$ where point $r_i$ is the location of node $v_i$
for $1\leq i\leq n$.
Let $\calR$ denote the set of all possible realizations.
We assume that the distributions of the locations of nodes in the metric space $\calP$ are independent, thus $\r$ occurs with probability $\Prob[\r]=\prod_{i\in [n]}\p_{v_i r_i}$, where $p_{vs}$ represents the probability that the location of node $v$ is point $s\in \calP$.
The model is also termed as the {\em locational uncertainty model} in \cite{kamousi2011stochastic}.
\item
(Existential Uncertainty Model)
A closely related model is the {\em existential uncertainty model} where the location of a node is a fixed point in the given metric space,
but the existence of the node is probabilistic.
In this model,
we use $\p_{i}$ to denote the probability that node $v_i$ exists (if exists, its location is $s_i$).
A realization $\r$  can be represented by a subset
$S\subset \calP$ and $\Prob[\r]=\prod_{s_i\in S}\p_i \prod_{s_i\notin S}(1-\p_i)$.
\end{enumerate}

\vspace{0.2cm}
\topic{Problem Formulation}
We are interested in following natural problem in the above models:
estimating the expected values of certain statistics of combinatorial objects.
In this paper, we study several combinatorial or geometry  problems in these two models:
the closest pair problem, minimum spanning tree,
minimum perfect matching (assuming an even number of nodes), $k$-clustering
and minimum cycle cover.
We take the minimum spanning tree problem for example.
Let
$\MST$ be the length of the minimum spanning tree (which is a random variable) and
$\MST(\r)$ be the length of the minimum spanning tree spanning all points in the realization $\r$.
We would like to estimate the following quantity:
$$
\E[\MST]=\sum_{\r\in \calR}\Prob[\r]\cdot \MST(\r).
$$
However, the above formula does not give us an efficient way to estimate the expectation since it involves an
exponential number of terms.
In fact, computing the exact expected value (for the problems considered in this paper) are either NP-hard or \#P-hard.
Following many of the theoretical computer science literatures on approximate counting
and estimation, our goal is to obtain fully polynomial randomized approximation schemes
for computing the expected values.

\subsection{Our Contributions}

We recall that a {\em fully polynomial randomized approximation scheme (FPRAS)} for a problem $f$
is a randomized algorithm $A$ that takes an input instance $x$, a real number $\e > 0,$
returns $A(x)$ such that
$ \Prob[(1-\epsilon)f(x)\le A(x)\le (1+\epsilon)f(x)]\ge\frac{3}{4}$
and its running time is polynomial in both the size of the input $n$ and $1/\e$.
Our main contributions can be summarized
in Table~\ref{tab:result}.
We need to explain some entries in the table in more details.

\begin{table*}[htbp]
  \centering
  \begin{tabular}{|*{4}{c|}}
    \hline
\multicolumn{2}{|c|}{Problems}
      &  Existential & Locational  \\\hline
\multirow{3}*{Closest Pair ($\S$\ref{sec:cp})}
      & $\Exp[\CP]$  &  $\FA$  & $\FA$  \\\cline{2-4}
      & $\Prob[\CP\leq 1]$ &  $\FA$  & $\FA$  \\\cline{2-4}
      & $\Prob[\CP\geq 1]$ &  $\inapp$  & $\inapp$  \\\hline
\multirow{3}*{Diameter ($\S$\ref{sec:cp})}
      & $\Exp[\D]$ &  $\FA$ & $\FA$   \\\cline{2-4}
      & $\Prob[\D\leq 1]$ &  $\inapp$ & $\inapp$  \\\cline{2-4}
      & $\Prob[\D\geq 1]$ &  $\FA$ & $\FA$   \\\hline
Minimum Spanning Tree ($\S$\ref{sec:mst}) & $\Exp[\MST]$ &  $\FA$\cite{kamousi2011stochastic}  & $\FA$  \\\hline
$k$-Clustering ($\S$\ref{sec:kcluster}) & $\Exp[\KC]$ & $\FA$   & $\open$  \\\hline
Perfect Matching ($\S$\ref{sec:mm}) & $\Exp[\MM]$ &  N.A. & $\FA$  \\\hline
$k$th Closest Pair ($\S$\ref{app:CP}) & $\Exp[\KCP]$ & $\FA$  &  $\open$  \\\hline
Cycle Cover  ($\S$\ref{app:cc})    & $\Exp[\CC]$ &  $\FA$  & $\FA$  \\\hline
$k$th Longest $m$-Nearest Neighbor ($\S$\ref{app:kmNN}) & $\Exp[\KMNN]$ & $\FA$  &  $\open$  \\\hline
  \end{tabular}
  \caption{Our results for some problems in different stochastic models.}
  \label{tab:result}
\end{table*}


\begin{enumerate}
\item Closest Pair:
We use $\CP$ to denote the minimum distance of any pair of two nodes.
If a realization has less than two nodes, $\CP$ is zero.
Computing $\Pr[\CP\leq 1]$ exactly in the existential model
is known to be \#P-hard even in an Euclidean plane~\cite{kamousi2014closest},
but no nontrivial algorithmic result is known before.
So is computing $\Pr[\CP\geq 1]$. In fact, it is not hard to show that
computing $\Pr[\CP\geq 1]$ is imapproximable within any factor in a metric space (Appendix~\ref{app:npcp}).

We also consider the problem of computing expected distance $\Exp[\CP]$ between the closest pair in the same model. We prove that the problem is \#P-hard in Appendix~\ref{app:npcp} and give the first known FPRAS in Section~\ref{sec:cp}. Note that an FPRAS for computing $\Pr[\CP\leq 1]$ does not imply an FPRAS for computing $\Exp[\CP]$
\footnote{
To the contrary, an FPRAS for computing $\Pr[\CP\geq 1]$ or $\Pr[\CP= 1]$ would imply an FPRAS for computing $\Exp[\CP]$
since $\Exp[\CP]=\sum_{(s_i,s_j)} \Pr[\CP=\dist(s_i,s_j)] \dist(s_i,s_j)
=\int \Pr[\CP\geq t] \d t=\sum_{(s_i,s_j)} \Pr[\CP\geq \dist(s_i,s_j)] (\dist(s_i,s_j)-\dist(s'_i,s'_j))$.
}.

\item
Diameter:
The problem of computing the expected length of the diameter can be reduced to the closest pair problem as follows.
Assume that the longest distance between two points in $\P$ is $W$.
We construct the new instance $\P'$ as follows: for any two points $u,v\in \P$,
let their distance be $2W-\dist(u,v)$ in $\P'$. The new instance is still a metric.
The sum of the distance of closest pair in $\P$ and the diameter in $\P'$ is exactly $2W$ (if there are at least two realized points).
Hence, the answer for the diameter can be easily derived from the answer for closest pair in $\P'$.

\item
Minimum Spanning Tree:
Computing $\Exp[\MST]$ exactly in both uncertainty models is known to be \#P-hard~\cite{kamousi2011stochastic}.
Kamousi, Chan, and Suri \cite{kamousi2011stochastic} developed an FPRAS for
estimating $\Exp[\MST]$
in the existential uncertainty model
and a constant factor approximation algorithm
in the locational uncertainty model.

Estimating $\Exp[\MST]$ is amendable to several techniques.
We obtain an FPRAS for estimating $\Exp[\MST]$ in the locational uncertainty model using the \core\ techinque in Section~\ref{sec:mst}.
In fact, the idea in \cite{kamousi2011stochastic} can also be extended to give an alternative FPRAS (Appendix~\ref{app:mst}).
It is not clear how to extend their idea to other problems.

\item
Clustering ($k$-clustering):
In the deterministic $k$-clustering problem,
we want to partition all points into $k$ disjoint subsets
such that the spacing of the partition is maximized,
where the spacing is defined to be the minimum of any
$\dist(u,v)$ with $u,v$ in different subsets \cite{kleinberg2006alg}.
In fact, the optimal cost of the problem is
the length of the $(k-1)$th most expensive edge in the minimum spanning tree \cite{kleinberg2006alg}.
We show how to estimate $\Exp[\KC]$ using the HPF (hierarchical partition family)
technique in Section~\ref{sec:kcluster}.


\item Perfect Matching:
We assume that there are even number of nodes
to ensure that a perfect matching always exists.
Therefore, only the locational uncertainty model is relevant here.
We give the first FPRAS for approximating the expected length of minimum perfect matching
in Section~\ref{sec:mm} using a more complicated \core\ technique.
\end{enumerate}

All of our algorithms run in polynomial time. However, we have not attempted to optimize
the exact running time.

\vspace{0.2cm}
\topic{Our techniques}
Perhaps the simplest and the most commonly used technique for estimating the
expectation of a random variable is the Monte Carlo method, that is
to use the sample average as the estimate.
However, the method is only efficient (i.e., runs in polynomial time)
if the variance of the random variable is small
(See Lemma~\ref{lm:chernoff}).
To circumvent the difficulty caused by the high variance,
a general methodology is to decompose the expectation of the random variable
into a convex combination of conditional expectations using the law of
total expectation:
$
\Exp[X]=\Exp_Y\big[\,\Exp[X\mid Y]\,\big]=\sum_{y}\Prob[Y=y]\,\Exp[X\mid Y=y].
$
Hopefully,  $\Prob[Y=y]$ can be estimated (or calculated exactly)
efficiently, and the random variable $X$ conditioning on each event $y$ has a low variance.
However, choosing the events $Y$ to condition on can be tricky.

We develop two new techniques
for choosing such events, each being capable of solving a subset of aforementioned problems.
In the first technique,
we first identify a set $\H$ of points,
called the {\em \core} of the problem,
such that (1): with high probability, all nodes realize in $\H$
and (2): conditioning on event (1), the variance is small.
Then, we choose $Y$
to be the number of nodes realized to points not in $\H$.
We compute the ($1\pm \epsilon$)-estimates for $Y=0,1$ using
Monte Carlo by (1) and (2).
The problematic part is when $Y$ is large, i.e., many nodes realize to points outside $\H$.
Even though the probability of such events is very small, the value of $X$ under such events may be considerably large,
thus contributing nontrivially.
However, we can show that the contribution of such events is
dominated by the first few events and thus can be safely ignored.
Choosing appropriate \core\ is easy for some problems, such as closest pair and minimum spanning tree,
while it may require additional idea for other problems such as minimum perfect matching.

Our second technique utilizes a notion called {\em Hierarchical Partition Family (HPF)}.
The HPF has $m$ levels, each representing a clustering of all points.
For a combinatorial problem, for which the solution is a set of edges,
we define $Y$ to be the highest level such that some edge in the solution is an
inter-cluster edge.
Informally, conditioning on the information of $Y$, we can essentially bound the variance of $X$
(hence use the Monte Carlo method).
To implement Monte Carlo, we need to be able to take samples efficiently conditioning on $Y$.
We show that such sampling problems can be reduced to, or have connections to, classical
approximate counting and sampling problems, such as approximating permanent, counting knapsack.

\subsection{Related Work}

Several geometric properties of a set of stochastic points
have been studied extensively in the literature
under the term {\em stochastic geometry}.
For instance, Bearwood et al.~\cite{beardwood1959shortest} shows that
if there are $n$ points uniformly and independently distributed in $[0, 1]^2$, the minimal traveling salesman
tour visiting them has
an expected length $\Omega(\sqrt{n})$.
Asymptotic results for minimum spanning trees and minimum matchings
on $n$ points uniformly distributed in unit balls
are established by Bertsimas and van Ryzin~\cite{bertsimas1990asymptotic}.
Similar results can be found in e.g.,
\cite{bern1993worst-case,karloff1989how,snyder1995aprioribounds}.
Compared with results in stochastic geometry, we focus on the efficient computation of the statistics,
instead of giving explicit mathematical formulas.

Recently, a number of researchers have begun to explore geometric computing under uncertainty
and many classical computational geometry problems have been studied
in different stochastic/uncertainty models.
Agarwal, Cheng, Tao and Yi~\cite{agarwal2009indexing} studied the problem of indexing probabilistic points with continuous distributions
for range queries on a line.
Agarwal, Efrat, Sankararaman, and Zhang~\cite{agarwal2012nearest}
also studied the same problem in the locational uncertainty model
under Euclidean metric.
The most probable $k$-nearest neighbor problem and its variants have attracted
a lot of attentions in the database community (See e.g.,~\cite{icde08-probnn}).
Several other problems have also been considered recently,
such as computing the expected volume of a set of probabilistic rectangles in a Euclidean space~\cite{yildiz2011union},
convex hulls \cite{agarwal2014convex},
skylines (Pareto curves) over probabilistic points~\cite{afshani2011approximate,atallah2011asymptotically},
and shape fitting~\cite{loffler2009shape}.

Kamousi, Chan and Suri~\cite{kamousi2011stochastic} initiated the study of estimating the expected length of combinatorial objects
in this model.
They showed that computing the expected length of
the nearest neighbor (NN) graph, the Gabriel graph (GG), the relative neighborhood graph (RNG),
and the Delaunay triangulation (DT) can be solved exactly in polynomial time, while
computing $\E[\MST]$ is \#P-hard and there exists a simple FPRAS for approximating $\Exp[\MST]$ in the existential model.
They also gave a deterministic PTAS for approximating $\Exp[\MST]$ in an Euclidean plane.
In another paper~\cite{kamousi2014closest}, they studied the closest pair
and (approximate) nearest neighbor problems (i.e., finding the point with the smallest expected distance from the query point)
in the same model.

The {\em randomly weighted graph} model
where the edge weights are independent nonnegative variables has also been studied extensively.
Frieze~\cite{frieze1985value} and Steele~\cite{steele1987frieze} showed that the expected value of the minimum spanning tree
on such a graph with identically and independently distributed edges is $\zeta(3)/D$
where $\zeta(3)=\sum_{j=1}^{\infty}1/j^3$ and $D$ is the derivative of the distribution at $0$.
Alexopoulos and Jacobson~\cite{alexopoulos2000state} developed algorithms that compute the distribution of
$\MST$ and the probability that a particular edge belongs to $\MST$
when edge lengths follow discrete distributions.
However, the running times of their algorithms may be exponential in the worst cases.
Recently, Emek, Korman and Shavitt~\cite{emek2011approximating} showed that
computing the $k$th moment of a class of properties, including the diameter, radius and minimum spanning tree,
admits an FPRAS for each fixed $k$.
Our model differs from their model in that the edge lengths are not independent.

The computational/algorithmic aspects of stochastic geometry have also gained a lot of attention in recent years from
the area of wireless networking. In many application scenarios, it is common to assume that
the nodes (e.g., sensors) are deployed randomly across a certain area, thereby forming a stochastic network.
It is of central importance to study various properties in this network, such as connectivity~\cite{gupta1998critical},
transmission capacity~\cite{gupta2000capacity}. We refer interested reader to a recent survey~\cite{haenggi2009stochastic}
for more references.

\eat{
A different model is mentioned by Bertsimas and Jaillet~\cite{jaillet1988apriori}. Like our model, it's non-uniform and considers the property of subsets of points. Moreover, so-called $universal~TSP$ is considered in~\cite{hajiaghayi2006improved,platzman1989spacefilling,shmoys2008aconstant}. It's a TSP such that for any subset of points, the cost is not too much longer compared to the optimal tour. The best known approximation algorithm is made up by David Shmoys and Kunal Talwar~\cite{shmoys2008aconstant}. Under the so-called "independent activation" model, they give the first constant approximation algorithm. Considering local property is also widely used in network structure. Many papers talk about it using different but similar model~\cite{althofer1993on,dhamdhere2005on,flaxman2005on,immorlica2004on,swamy2006approximation}.
}

\subsection{Preliminaries}

Before describing our main results,
we first consider the straightforward Monte Carlo strategy,
which is an important building block in our later developments.
Suppose we want to estimate $\Exp[X]$.
In each Monte Carlo iteration, we take a sample (a realization of all nodes),
and compute the value of $X$ for the sample.
At the end, we output the average over all samples.
The number of samples required by this algorithm is suggested by the following
standard Chernoff bound.
\begin{lemma} {\em (Chernoff Bound)}
\label{lm:chernoff}
Let random variables $X_1, X_2, \ldots, X_N$ be independent random variables taking on values between 0 and $U$.
Let $X = \frac{1}{N}\sum_{i=1}^N X_i$ and $\mu$ be the expectation of $X$, for any $\e > 0$,
$$
\Pr \left[ X \in [(1-\e)\mu, (1+\e)\mu] \right] \geq 1-2e^{-N \frac{\mu}{U} \e^2 /4}.
$$
\end{lemma}
Therefore, for any $\e>0$, in order to get an $(1\pm\e)$-approximation with probability $1-\frac{1}{\poly(n)}$,
the number of samples  needs to be $O(\frac{U}{\mu \e^2}\log n)$.
If $\frac{U}{\mu}$, the ratio between the maximum possible value of $X$ and the expected value $\E[X]$,
is bounded by $\poly(m,n,\frac{1}{\epsilon})$,
we can use the above Monte Carlo method to estimate $\E[X]$ with
a polynomial number of samples. Since we use this condition often, we devote a separate definition to it.

\begin{definition}
We call a random variable $X$ {\em poly-bounded} if
the ratio between the maximum possible value of $X$ and the expected value $\E[X]$
is bounded by $\poly(m,n,\frac{1}{\epsilon})$.
\end{definition}


\section{The Closest Pair Problem}
\label{sec:cp}

\subsection{Estimating $\Pr[\CP\leq 1]$}
\label{sec:warmup}

As a warmup, we first demonstrate how to use the \core\ technique for the closest pair problem in the existential uncertainty model.
Given a set of points $\P=\{s_1, \ldots, s_m\}$ in the metric space, where each point $s_i\in \P$ is present with probability $p_i$.
We use $\CP$ to denote the distance between the closest pair of vertices in the realized graph.
If the realized graph has less than two points, $\CP$ is zero.
The goal is to compute the probability $\Prob[\CP\leq 1]$.

For a set $H$ of points and a subset $S\subseteq H$, we use $H\A{S}$ to denote the event that among all points in $H$, all and
only points in $S$ are present.
For any nonnegative integer $i$, let $H\A{i}$ to denote the event $\bigvee_{S\subseteq H:|S|=i}H\A{S}$, i.e., the event that exactly $i$ points are present in $H$.

The {\em\core} of the closest pair problem is simply defined to be
$$
\H=\left\{s_i\mid p_i\geq \frac{\e}{m^2}\right\}.
$$
Let $\calF=\P\setminus \H$.
We consider the decomposition
$$
\Prob[\CP\leq 1]=\sum_{i=0}^{|\calF|}\Prob[\calF\A{i} \wedge \CP\leq 1]=\sum_{i=0}^{|\calF|}\Prob[\calF\A{i}]\cdot \Prob[\CP\leq 1\mid \calF\A{i}].
$$
Our algorithm is very simple:
estimate the first three terms (i.e., $i=0,1,2$) and use their sum as our final answer.
\vspace{0.2cm}

We can see that $\H$ satisfies the two properties of a \core\
mentioned in the introduction:
\begin{enumerate}
\item The probability that all nodes are realized in $\H$, i.e., $\Pr[\calF\A{0}]$, is at least $1-m\cdot \frac{\e}{m^2}=1-\frac{\e}{m}$;
\item If there exist two points $s_i,s_j\in \H$ such that $\dist(s_i,s_j)\leq 1$, we have $\Prob[\,\CP\leq 1\mid \calF\A{0}\,]\geq \frac{\e^2}{m^4}$; otherwise, $\Prob[\CP\leq 1\mid \calF\A{0}]=\Prob[\H\A{0}\mid \calF\A{0}]+\Prob[\H\A{1}\mid \calF\A{0}]$. Note that we can compute $\Prob[\H\A{0}\mid \calF\A{0}]$ and $\Prob[\H\A{1}\mid \calF\A{0}]$ in polynomial time.
\end{enumerate}
Both properties guarantee that the random variable $I(\CP\leq 1)$,
conditioned on $\calF\A{0}$, is poly-bounded
\footnote
{
$I()$ is the indicator function. Note that $\Exp[I(\CP\leq 1)]=\Pr[\CP\leq 1]$.
},
hence we can easily get a ($1\pm\e$)-estimation for $\Prob[\calF\A{0}\wedge \CP\leq 1]$ with polynomial many samples
with high probability.
Similarly, $\Prob[\calF\A{i}\wedge \CP\leq 1]$ can also be estimated
with polynomial number of samples for $i=1,2$.
The algorithm can be found in Algorithm~\ref{algo:cp}.

\linesnumbered
\begin{algorithm}[h]
\caption{Estimating $\Pr[\CP\leq 1]$}
\label{algo:cp}
Estimate $\Prob[\calF\A{0} \wedge \CP\leq 1]$:
Take $N_0=O\bigl((m/\e)^{4}\ln m \bigr)$ independent samples. Suppose $M_0$ is the number of samples satisfying $\CP\leq 1$ and $\calF\A{0}$.
$T_0\leftarrow \frac{M_0}{N_0}$. \\
Estimate $\Prob[\calF\A{1} \wedge \CP\leq 1]$: For each point $s_i\in \calF$, take $N_1=O((m/\e)^{4}\ln m)$ independent samples conditioning on the event $\calF\A{\{s_i\}}$. Suppose there are $M_i$ samples satisfying $\CP\leq 1$.
$T_1\leftarrow \sum_{s_i\in \calF}p_iM_i/N_1$.\\
Estimate $\Prob[\calF\A{2} \wedge \CP\leq 1]$: For each point pair $s_i,s_j\in \calF$, take $N_2=O((m/\e)^{4}\ln m)$ independent samples conditioning on the event
$\calF\A{\{s_i,s_j\}}$. Suppose there are $M_{ij}$ samples satisfying $\CP\leq 1$.
$T_2\leftarrow \sum_{s_i,s_j \in \calF}p_i p_jM_{ij}/N_2$.\\
\KwSty{Output:} $T_0+T_1+T_2$
\end{algorithm}

\begin{lemma}
\label{lm:cp}
	Steps 1,2,3 in Algorithm~\ref{algo:cp} provide
	$(1\pm \epsilon)$-approximations for $\Prob[\calF\A{i} \wedge \CP\leq 1]$
	for $i=0,1,2$ respectively,
	with high probability.
\end{lemma}

\eat{
For estimating $\Prob[\calF\A{2}\wedge \CP \leq 1]$, we rewrite this term by $\sum_{s_i,s_j\in \calF}\Prob[\calF\A{\{s_i,s_j\}}\wedge \CP \leq 1]$. For two points $s_i,s_j\in \calF$, note that $\Prob[\calF\A{\{s_i,s_j\}}\wedge \CP \leq 1]=\Prob[\calF\A{\{s_i,s_j\}}]\cdot \Prob[\CP \leq 1\mid \calF\A{\{s_i,s_j\}}]$. Similarly, we can use $p_i p_j$ to estimate $\Prob[\calF\A{\{s_i,s_j\}}]$. If $d(s_i,s_j)\leq 1$, we have $\Prob[\CP \leq 1\mid \calF\A{\{s_i,s_j\}}]=1$. Otherwise, we can use the similar argument as in estimating $\Prob[\calF\A{1}\wedge \CP \leq 1]$.
}

\begin{theorem}
There is an FPRAS for estimating the probability of the distance between
the closest pair of nodes is at most $1$ in the existential uncertainty model.
\end{theorem}

\begin{proof}
We only need to show that the contribution from the rest of terms
(where more than three points outside \core\ $\H$ are present)
is negligible compared to the third term.
Suppose $S$ is the set of all present points such that $\CP\leq 1$
and there are at least 3 points not in $\H$.
Suppose $s_i,s_j$ are the closest pair in $S$.
We associate $S$ with a smaller set $S'\subset S$ by making 1 present point in $(S\cap \calF)\setminus\{s_i,s_j\}$ absent
(if there are several such $S'$, we choose an arbitrary one).
We denote it as $S\consistent S'$.
We use the notation $S\in F_i$ to denote that the realization $S$
satisfies  $(\calF\A{i}\wedge \CP\leq 1)$.
Then, we can see that for $i\geq 3$,
\begin{align*}
\Prob[\calF\A{i}\wedge \CP\leq 1]
& =\sum_{S: S\in F_i}\Prob[S]
\leq\sum_{S': S'\in F_{i-1}} \sum_{S: S\consistent S'} \Prob[S].
\end{align*}
For a fixed $S'$, there are at most $m$ different sets $S$ such that $S\consistent S'$
and $\Prob[S]\leq \frac{2\e}{m^2}\Prob[S']$ for any such $S$.
Hence, we have that
$$
\sum_{S: S\consistent S'} \Prob[S]\leq \frac{2\e}{m}\Prob[S'].
$$
Therefore,
$$
\Prob[\calF\A{i}\wedge \CP\leq 1]
\leq \frac{2\e}{m}\cdot \sum_{S': S'\in F_{i-1}} \Prob[S']
= \frac{2\e}{m}\cdot  \Prob[\calF\A{i-1}\wedge \CP\leq 1].
$$
Hence, overall we have
$
\sum_{i\geq 3}\Prob[\calF\A{i}\wedge \CP\leq 1] \leq \e \Prob[\calF\A{2}\wedge \CP\leq 1].
$
This finishes the analysis. \\
\qed
\end{proof}

Note that the number of samples is dominated by estimating $\Prob[\calF\A{2} \wedge \CP\leq 1]$. Since there are $O\bigl(m^2\bigr)$ different pairs $s_i,s_j\in\calF$. We take $N_2$ independent samples for each pair. Overall, we take $O\bigl(\frac{m^6}{\epsilon^4}\ln m\bigr)$ independent samples.

\topic{Locational Uncertainty Model}
The algorithm for the locational uncertainty model is similar to the one for the existential uncertainty model.
Here we briefly sketch the algorithm.
For ease of exposition, we assume that for each point, there is only one node that may be realized
at this point.
In principle, if more than one node may be realized at the same point, we can
create multiple copies of the point co-located at the same place.

For any node $v\in \V$ and point $s\in \P$,
we use the notation $v \realize s$ to denote the event that
node $v$ is realized at point $s$.
Let $\p_{vs}=\Prob[v \realize s]$, i.e., the probability that node $v$ is realized at point $s$. For each point $s\in \P$, we let $p(s)$ denote
the probability that point $s$ is present ($p(s)=p_{vs}$, $v$ is the unique node which may be realized at $s$).
Let $H\A{i}$ denote the event that exactly $i$ nodes are realized to the point set $H$.

We construct the \core\ $\H=\{s\mid p(s)\geq \frac{\epsilon}{(nm)^2}\}$.
Let $\calF=\P\setminus \H$.
Then we rewrite $\Prob[\CP \leq 1]=\sum_{0\leq i\leq n}\Prob[\calF\A{i}\wedge \CP \leq 1]$.
We only need to estimate the first three terms.

\vspace{0.3cm}
\topic{Estimating $\Prob[\calF\A{0}\wedge \CP \leq 1]$}
\begin{enumerate}
\item If there exist two points $s,t\in \H$ with $\dist(s,t)\leq 1$ which correspond to different nodes, then $\Prob[\calF\A{0}\wedge \CP \leq 1]\geq p(s)p(t)\geq \frac{\epsilon^2}{(nm)^4}$ by the definition of \core\,, we can simply estimate $\Prob[\calF\A{0}\wedge \CP\leq 1]$ by taking $O(\frac{(nm)^4}{\epsilon^4}\ln m)$ independent samples using the Monte Carlo method.
\item If no such two points $s,t\in \H$ exist, $\Prob[\calF\A{0}\wedge \CP \leq 1]=0$.
\end{enumerate}

\eat{
\begin{lemma}
\label{lm:Asmall}
$A\leq \epsilon \sum_{s,t\in \H, \dist(s,t)\leq 1}\Prob[\realize s \wedge \realize t]$.
\end{lemma}
\begin{proof}
Note that $s,t,s',t'$ correspond to 4 distinct nodes, then
$\Prob\left[\bigwedge_{x\in \{s, t, s',t'\}}\realize x\right]=p(s)p(t)p(s')p(t')$.
Otherwise, it is zero. Hence, we can easily get that
$$
\sum_{\overset{s,t,s',t'\in \H,}{\dist(s,t)\leq 1,\dist(s',t')\leq 1}}
    \Prob\left[\bigwedge_{x\in \{s, t, s',t'\}}\realize x\right]
    \leq \frac{\epsilon^3}{(nm)^6}\cdot m^2\cdot \sum_{s,t\in \H, \dist(s,t)\leq 1}\Prob[\realize s\wedge \realize t]
$$
Similarly, we have that
$$
\sum_{\overset{s,t,t'\in \H,}{\dist(s,t)
    \leq 1,\dist(s,t')\leq 1}}
    \Prob\left[ \bigwedge_{x\in \{s,t, t'\}}\realize x\right]
    \leq \frac{\epsilon}{(nm)^2}\cdot m\cdot \sum_{s,t\in \H, \dist(s,t)\leq 1}\Prob[\realize s\wedge \realize t]
$$
Combine the two inequalities, we have $A\leq \epsilon \sum_{s,t\in \H, \dist(s,t)\leq 1}\Prob[\realize s \wedge \realize t]$.
\qed
\end{proof}
}

\vspace{0.3cm}
\topic{Estimating $\Prob[\calF\A{1}\wedge \CP \leq 1]$}
We first rewrite this term by $\sum_{v\in \V, s\in \calF}\Prob[\calF\A{1}\wedge \CP \leq 1 \wedge v\realize s]$.
For a node $v\in \V$ and point $s\in \calF$,
we denote $\B_s=\{t\in \calH: \dist(s,t)\leq 1\}$.
If $\B_s$ contains any point corresponding to a node other than $v$,
we can use Monte Carlo for estimating
$\Prob[\calF\A{1}\wedge \CP \leq 1 \mid v\realize s]$ since it is at least $\frac{\epsilon}{(nm)^2}$.
Otherwise,
computing $\Prob[\calF\A{1}\wedge \CP \leq 1 \mid v\realize s]$
is equivalent to computing
$\Prob[\calF\A{0}\wedge \CP \leq 1]$
in the instance without $v$ (since $v$ is at distance more than 1 from any other nodes).

\vspace{0.3cm}
\topic{Estimating $\Prob[\calF\A{2}\wedge \CP \leq 1]$}
We rewrite it as $\sum_{v,v'\in \V, s,s'\in \calF}\Prob[\calF\A{2}\wedge \CP \leq 1 \wedge v\realize s \wedge v'\realize s']$.
We estimate each term in the same way as the former case. We do not repeat the argument here.

\vspace{0.3cm}
\topic{Analysis}
Similar to the existential uncertainty model,
we can show that the contribution of $\sum_{3\leq i\leq n}\Prob[\calF\A{i}\wedge \CP \leq 1]$ is negligible.
The argument is almost the same as before.
Suppose $S$ is a realization such that $\CP\leq 1$ and there are at least 3 points not in $\H$.
Suppose $v_i,v_j$ are the closest pair in $S$.
We associate $S$ with $S'$, where $S'$ is obtained by sending node $v$ in $S$
(except $v_i,v_j$) located in $\calF$ to a point
$s\in \H$ such that $\p_{vs}\geq \frac{1}{2m}$.
We denote it as $S\consistent S'$.
Then for a fixed $S'$, there are at most $nm$ different sets $S$ such that $S\consistent S'$
and $\Prob[S]\leq \frac{2\epsilon}{n}\Prob[S']$ for any such $S$. The rest arguments are the same.

\begin{theorem}
There is an FPRAS for estimating the probability of the distance between
the closest pair of nodes is at most $1$ in the locational uncertainty model.
\end{theorem}

The number of samples is dominated by estimating $\Prob[\calF\A{2} \wedge \CP\leq 1]$.
Since there are $O\bigl(n^2\bigr)$ different pairs of nodes $v,v'\in \V$
and $O\bigl(m^2\bigr)$ different pairs of points $s,s'\in \calF$,
we separate $\calF\A{2}$ into $O\bigl(n^2m^2\bigr)$ different terms. For each term, we take $O\bigl(\frac{(nm)^4}{\e^4}\ln m\bigr)$ independent samples. Thus, we take $O\bigl(\frac{n^6m^6}{\epsilon^4}\ln m \bigr)$ independent samples in total.

\vspace{-0.3cm}
\subsection{Estimating $\Exp[\CP]$}
\label{sec:ecp}
In this section, we consider the problem of estimating $\Exp[\CP]$,
where $\CP$ is the distance of the closest pair of present points,
in the existential uncertainty model.
Now, we introduce our second main technique,
the {\em hierarchical partition family (HPF)} technique, to solve this problem.
An HPF is a family $\Psi$ of partitions of $\calP$, formally defined as follows.


\eat{
We also note that an FPRAS for computing $\Pr[\CP\leq 1]$ does not imply an FPRAS for computing $\Exp[\CP]$, hence
the algorithm developed in the previous section can not be used here
\footnote{
To the contrary, an FPRAS for computing $\Pr[\CP\geq 1]$ or $\Pr[\CP= 1]$ would imply an FPRAS for computing $\Exp[\CP]$
since $\Exp[\CP]=\sum_{(s_i,s_j)} \Pr[\CP=\dist(s_i,s_j)] \dist(s_i,s_j)
=\int \Pr[\CP\geq t] \d t=\sum_{(s_i,s_j)} \Pr[\CP\geq \dist(s_i,s_j)] (\dist(s_i,s_j)-\dist(s'_i,s'_j))$.
(THIS SHOULD BE IN INTRODUCTION)
}.
}

\begin{definition}
\label{def:hpf}
(Hierarchical Partition Family (HPF))
Let $T$ be any minimum spanning tree spanning all points of $\P$.
Suppose that the edges of $T$ are $e_1,\ldots,e_{m-1}$
with $\dist({e_1})\geq \dist({e_2})\geq \ldots \geq \dist({e_{m-1}})$.
Let $E_i=\{e_i,e_{i+1},\ldots,e_{m-1}\}$.
The HPF $\hpf$ consists of $m$ partitions $\Gamma_1,\ldots, \Gamma_m$.
$\Gamma_1$ is the entire point set $\calP$.
$\Gamma_i$ consists of $i$ disjoint subsets of $\calP$,
each corresponding to a connected component of $G_i=G(\calP,E_{i})$.
$\Gamma_m$ consists of all singleton points in $\calP$.
It is easy to see that $\Gamma_j$ is a refinement of $\Gamma_i$ for $j>i$.
Consider two consecutive partitions $\Gamma_i$ and $\Gamma_{i+1}$.
Note that $G_{i}$ contains exactly one more edge (i.e., $e_i$) than $G_{i+1}$.
Let $\mu'_{i+1}$ and $\mu''_{i+1}$ be the two components (called the {\em split components}) in $\Gamma_{i+1}$, each containing an endpoint of $e_i$.
Let $\nu_i\in \Gamma_i$ be the connected component of $G_i$ that contains $e_i$.
We call $\nu_i$ the {\em special component} in $\Gamma_i$. Let $\Gamma'_i=\Gamma_i\setminus \nu_i$.
\end{definition}

We observe two properties of $\hpf$ that are useful later.
\begin{enumerate}
\item[P1.] Consider a component $C\in \Gamma_i$.
Let $s_1,s_2$ be two arbitrary points in $C$.
Then $\dist(s_1,s_2)\leq (m-1) \dist(e_i)$
(this is because $s_1$ and $s_2$ are connected in $G_i$, and $e_i$ is the longest edge in $G_i$).
\item[P2.] Consider two different components $C_1$ and $C_2$ in $\Gamma_i$.
Let $s_1\in C_1$ and $s_2\in C_2$ be two arbitrary points.
Then $\dist(s_1,s_2)\geq \dist(e_{i-1})$
(this is because the minimum inter-component distance is $\dist(e_{i-1})$ in $G_i$).
\end{enumerate}


Let the random variable $Y$ be smallest integer $i$ such that
there is at most one present point in each component of $\Gamma_{i+1}$.
Note that if $Y=i$ then
each component of $\Gamma_{i}$ contains at most one point, except that
the special component $\nu_i$ contains exactly two present points.
The following lemma is a simple consequence of P1 and P2.

\begin{lemma}
\label{lm:closepair}
Conditioning on $Y=i$, it holds that $\dist(e_{i})\leq \C\leq m \dist(e_{i})$ 
(hence, $\C$ is poly-bounded).
\end{lemma}

Consider the following expansion of $\Exp[\CP]$:
$$
\Exp[\CP]=\sum_{i=1}^{m-1} \Prob[Y=i]\Exp[\CP\mid Y=i].
$$
For a fixed $i$, $\Prob[Y=i]$ can be estimated as follows:
For a component $C\subset\calP$, we use $C\A{j}$ to denote the event that exactly $j$ points in $C$ are present,
$C\A{s}$ the event that only $s$ is present in $C$ and
$C\A{\leq j}$  the event that no more than $j$ points in $C$ are present.
Let $\mu'_{i}$ and $\mu''_{i}$ be the two split components in $\Gamma_{i}$.
Note that
$$
\Prob[Y=i]=\Prob[\mu'_{i+1}\A{1}]\cdot \Prob[\mu''_{i+1}\A{1}]\cdot \prod_{C\in \Gamma'_i}\Prob[C\A{\leq 1}].
$$
Each term can be easily computed in polynomial time.
The remaining is to show how to estimate $\Exp[\CP\mid Y=i]$.
Since $\C$ is poly-bounded, it suffices to give an efficient algorithm
to take samples conditioning on $Y=i$.
This is again not difficult:
We take exactly one point $s\in \mu'_{i+1}$ with probability $\Pr[\mu'_{i+1}\A{s}]/\Pr[\mu'_{i+1}\A{1}]$.
Same for $\mu''_{i+1}$.
For each $C\in \Gamma'_i$,
take no point from $C$ with probability $\Pr[C\A{0}]/\Pr[C\A{\leq 1}]$;
otherwise, take exactly one point $s\in C$ with probability $\Pr[C\A{s}]/\Pr[C\A{\leq 1}]$.

By Lemma~\ref{lm:closepair}, conditioning on $Y=i$, taking $O(\frac{m}{\epsilon^2}\ln m)$ independent samples are enough using the Monte Carlo method. Since there are $m$ levels, we take $O\bigl(\frac{m^2}{\epsilon^2}\ln m \bigr)$ independent samples in total.
This finishes the description of the FPRAS in the existential uncertainty model.

\vspace{0.3cm}
\topic{Locational Uncertainty Model}
Our algorithm is almost the same as the existential model.
We first construct the HPF $\hpf$.
The random variable $Y$ is defined in the same way.
The only difference is how to estimate $\Prob[Y=i]$ and
how to take samples efficiently conditioning on $Y=i$.
First consider estimating $\Prob[Y=i]$.
We can consider the problem as the following bins-and-balls problem:
we have $n$ balls (corresponding to nodes) and $i$ bins (corresponding to components in $\Gamma_i$).
Each ball $v$ is thrown to bin $C$ with probability $p_{vC}=\sum_{s\in C}p_{vs}$ (note that $\sum_C p_{vC}=1$).
We want to compute the probability that
each of the first and second bins (corresponding to the two split components) contains exactly one ball,
and for other bins each contains at most one ball.
Consider the following $i\times i$ ($i\geq n$) matrix $M$ with
$
M_{vC}=\left\{
         \begin{array}{ll}
           p_{vC}=\sum_{s\in C}p_{vs}, & \hbox{for $v\in [n]$ and $C\in [i]$;} \\
           1, & \hbox{otherwise}
         \end{array}
       \right.
$.
It is not difficult to see that the permanent
$$
\per(M)=\sum_{\sigma\in \mathbb{S}_i}\prod_{v} M_{v\sigma(v)}
$$
is exactly the probability that each bin contains at most one ball.
To enforce each of the first two bins contains exactly one ball,
simply consider the Laplace expansion of $\per(M)$, expanded along the first two columns,
and retain those relevant terms:
$$
\Pr[Y=i]=\sum_{k\in [n]}\sum_{j\in [n],j\ne k} M_{k1}M_{j2} \per(M^\star_{kj})
$$
where $M^\star_{kj}$ is $M$ with the 1st and 2nd columns and $k$th and $j$th rows removed.
Then, we can use the celebrated result for approximating permanent by Jerrum, Sinclair, and Vigoda~\cite{jerrum2004permanent}
to get an FPRAS for approximating $\Pr[Y=i]$.
In fact, the algorithm in \cite{jerrum2004permanent} provides
a fully polynomial time approximate sampler for perfect matchings
\footnote{
The approximate sampler can return in poly-time
a permutation $\sigma\in \mathbb{S}_i$ with probability $(1\pm \epsilon)\prod_{s} M_{s\sigma(s)}/\per(M)$.
}. This can be easily translated to an efficient sampler conditioning on $Y=i$
\footnote{
We can also use the generic reduction by Jerrum, Valiant and Vazirani \cite{jerrum1986random}
which can turn an FPRAS into a poly-time approximate sampler for self-reducible relations.
}.
Finally, we remark that the above algorithm can be easily modified to handel the case
with both existential and locational uncertainty model.

\eat{
\vspace{0.3cm}
\linesnumbered
\begin{algorithm}[t]
\caption{Estimating $\Exp[\CP]$}
\label{alg:estCP}
\KwSty{Input:} The point set $\P$, and the existence probability $p_i$ for each point $s_i$.\\
For each $1\leq i\leq m-1$, take $N=O(\frac{m^5}{\e^5})$ independent random samples $G_{i,j}$ conditioning on $L_i$.\\
For a sample of level $i$, independently sample 1 point $s\in T^*_{i,1}$ with probability
$\frac{\Prob[T^*_{i,1}\A{s}]}{\Prob[T^*_{i,1}\A{1}]}$ and 1 point $t\in T^*_{i,2}$ with probability $\frac{\Prob[T^*_{i,2}\A{t}]}{\Prob[T^*_{i,2}\A{1}]}$. For other components $T_{i,k}\in \calT_i$, sample no point $s\in T_{i,k}$ with probability $\frac{\Prob[T_{i,k}\A{0}]}{\Prob[T_{i,k}\A{\leq 1}]}$ and sample 1 point $s\in T_{i,k}$ with probability $\frac{\Prob[T_{i,k}\A{s}]}{\Prob[T_{i,k}\A{\leq 1}]}$.\\
For $1\leq i\leq m-1$, $C_i\leftarrow \frac{1}{N}\sum_{1\leq j\leq N}\CP(G_{i,j})$.\\
\KwSty{Output:} $\sum_{i=1}^{m-1} \Prob[L_i] \cdot C_i$
\end{algorithm}
}

\begin{theorem}
\label{thm:cp}
There is an FPRAS for estimating the expected distance between
the closest pair of nodes in both existential and locational uncertainty models.
\end{theorem}

\eat{
\begin{proof}
We only need to show that $C_i$ is a ($1\pm\e$)-estimation for $\Prob[L_i]\Exp[\CP\mid L_i]$ with high probability. By construction, conditioning on $L_i$, the minimum value of $\CP$ is at least $\dist(e_{i})$ while the maximum possible value is at most $m\cdot \dist(e_{i})$.
Thus, $\CP$ is poly-bounded since that the sample average $C_i$ is a good estimation of $\Exp[\CP\mid L_i]$.

The remaining task is to show sample $G_{i,j}$ is a uniformly random sample conditioning on $L_i$. In fact, it is no hard to see since nodes in $T^*_{i,1}\cup T^*_{i,2} \cup \calT_i$ are independent (no two components share a common point).
\qed
\end{proof}

The algorithm in the locational uncertainty model can be found in Appendix \ref{app:ECP}.

\eat{
\subsection{Diameter}
In addition, we show that the problem of computing the expected length of the diameter is the dual problem of computing $\Exp[\KC]$ under the existence uncertainty model. Assume the longest distance between two points in $\P$ is $W$. We construct $\P'$ as the following: for any two points $u,v\in \P$ of distance $\dist(u,v)$, change the distance between them into $2W-\dist(u,v)$ in $\P'$. Then for a sample $G$ maintaining at least two points, the sum of $\CP(G)$ in $\P$ and the diameter of $G$ in $\P'$ is exactly $2W$. Since we could produce a $(1\pm\epsilon)$-estimate for $\Exp[\CP]$, an $\FA$ for $\Exp[\diam]$ could also be achieved. Using the same method, we could achieve the result of the diameter in the Table~\ref{tab:result}.
}
}
\topic{$k$th Closest Pair}
In addition, we consider the problem of the expected distance $\Exp[\KCP]$ between the $k$th closest pair
under the existential uncertainty model. We use the HPF technique, and construct an efficient sampler via a dynamic programming. The details can be found in Appendix~\ref{app:CP}.

\section{$k$-Clustering}
\label{sec:kcluster}

In this section, we study the k-clustering problem in the existential uncertainty model.
According to~\cite{kleinberg2006alg}, the optimal objective value for $k$-clustering
is the $(k-1)$th most expensive edge of the minimum spanning tree.
We consider estimating $\Exp[\KC]$ under the existential uncertainty model.

Denote the point set $\P=\{s_1, \ldots, s_m\}$, where each point $s_i\in \P$ is present with probability $p_i$.
We construct the HPF $\hpf$.
%
Let the random variable $Y$ be the largest integer $i$ such that at most $k-1$ components in $\Gamma_i$ contain at least one present point. Let $\Gamma'_i=\Gamma_i\setminus \nu_i$.
Note that if $Y=i$ then at most $k-2$ components in $\Gamma'_i$ contain present points while the special component $\nu_i$ contains at least two present points, since both component $\mu'_{i+1}$ and $\mu''_{i+1}$ contain at least one present point. By the property P1 and P2 of HPF, we have the following lemma.

%
%

\vspace{-0.2cm}
\begin{lemma}
\label{lm:kcluster}
Conditioning on $Y=i$, it holds that $\dist(e_{i})\leq \KC\leq m \dist(e_{i})$
(hence, $\KC$ is poly-bounded)..
\end{lemma}
\vspace{-0.2cm}
\begin{proof}
Since $\Gamma_{i+1}$ contains at least $k$ nonempty components, any
spanning tree must have at least $k-1$ inter-component edges.
Any inter-component edge is of length at least $\dist(e_i)$,
so is the $(k-1)$th expensive edge.
Now we show the other direction. Assume w.l.o.g. that all pairwise distances are distinct.
Consider a realization satisfying $Y=i$
and the graphical matroid
which consists of all forests of the realization.
Suppose $\KC=\dist(e)$ for some edge $e$. Let $E_e$ be all edges with length no larger than $e$
in this realization. We can see that $\rank(E_e)=n-k+1$ where $\rank$ is the matroid rank function and
$n$ the number of present points in the realization.
Hence, any spanning tree contains no more than $n-k+1$ edges from $E_e$.
Equivalently, the $(k-1)$th most expensive edge of any spanning tree is no smaller than $\KC$.
Moreover, since $\Gamma_i$ has no more than $k-1$ nonempty components,
there exists a spanning tree such that the $(k-1)$th most expensive edge
is an intra-component edge in $\Gamma_i$. The lemma follows from P1.
\qed
\end{proof}

Consider the following expansion
$
\Exp[\KC]=\sum_{i=1}^{m-1} \Prob[Y=i]\Exp[\KC\mid Y=i].
$
Recall that for a component $C\subset\calP$, we use $C\A{j}$ to denote the event that exactly $j$ points in $C$ are present,
$C\A{s}$ the event that only $s$ is present in $C$ and
$C\A{\leq j}$ $(C\A{\geq j})$  the event that at most (at least) than $j$ points in $C$ are present. For a partition $\Gamma$ on $\P$, we use $\Gamma\A{j,\geq 1}$ to denote the event that exactly $j$ components in $\Gamma$ contain at least one present point. Note that
$$
\Prob[Y=i]=\Prob[\mu'_{i+1}\A{\geq 1}]\cdot \Prob[\mu''_{i+1}\A{\geq 1}]\cdot \Prob[\Gamma'_i\A{k-2,\geq 1}].
$$
Note that $\Prob[\mu'_{i+1}\A{\geq 1}]$ and $\Prob[\mu''_{i+1}\A{\geq 1}]$ can be easily computed in polynomial time.
The remaining task is to show how to compute $\Prob[\Gamma'_i\A{k-2,\geq 1}]$ and how to estimate $\Exp[\KC\mid Y=i]$.
We first present a simple lemma which is useful later.
\begin{lemma}
\label{lm:samTj}
For a component $C$ and $j\in \mathbb{Z}$,
we can compute $\Prob[C\A{j}]$ (or $\Prob[C\A{\geq j}]$) in polynomial time.
Moreover, there exists a poly-time sampler to sample present points from $C$ conditioning on $C\A{j}$ (or $C\A{\geq j}$).
\end{lemma}
\begin{proof}
The idea is essentially from \cite{dyer2003approximate}.
W.l.o.g, we assume that the points in $C$ are $s_1,\ldots, s_{n}$.
We denote the event that among the first $a$ points,
exactly $b$ points are present by $E[a,b]$ and denote the probability of $E[a,b]$ by $\Prob[a,b]$.
Note that our goal is to compute $\Pr[n,j]$, which can be solved by the following dynamic program:
\begin{enumerate}
\item If $a < b$, $\Prob[a,b]=0$. If $a=b$, $\Prob[a,b]=\prod_{1\leq l\leq a}p_l$. If $b=0$, $\Prob[a,b]=\prod_{1\leq l\leq a}(1-p_l)$.
\item For $a>b$ and $b\geq 1$, $\Prob[a,b]=p_a \Prob[a-1,b-1]+(1-p_a) \Prob[a-1,b]$.
\end{enumerate}
We can also use this dynamic program to construct an efficient sampler. Consider the point $s_{n}$.
With probability $p_{n} \Prob[n-1,j-1]/\Prob[n,j]$, we make it present and then recursively consider the point $s_{n-1}$ conditioning on the event $E[n-1,j-1]$.
With probability $(1-p_{n}) \Prob[n-1,j]/\Prob[n,j]$, we discard it and then recursively sample conditioning on the event $E[n-1,j]$.
$\Pr[C\A{\geq j}]$
can be handled in the same way and we omit the details.\qed
\end{proof}

%

\vspace{0.1cm}
\topic{Computing $\Prob[\Gamma'_i\A{k-2,\geq 1}]$}
Now, it is ready to show how to compute $\Prob[\Gamma'_i\A{k-2,\geq 1}]$ in polynomial time.
Note that for each component $C_j\in \Gamma'_i$, we can easily compute $q_j=\Pr[C_j\A{\geq 1}]$ in polynomial time.
Since all components in $\Gamma'_i$ are disjoint,
using Lemma~\ref{lm:samTj} (consider each component $C_j$ in $\Gamma'_i$ as a point with existential probability $q_j$),
we can compute $\Prob[\Gamma'_i\A{k-2,\geq 1}]$.

\vspace{0.3cm}
\noindent
To take samples conditioning on $Y=i$, we first sample $k-2$ components in $\Gamma'_i$ which contain present points.
Then for these $k-2$ components and $\mu'_{i+1}$, $\mu''_{i+1}$, we independently sample present points in each component using Lemma~\ref{lm:samTj}.
By Lemma~\ref{lm:kcluster}, for estimating $\Exp[\KC\mid Y=i]$, we need to take $O\bigl(\frac{m}{\e^2}\ln m\bigr)$ independent samples. So we take $O\bigl(\frac{m^2}{\e^2}\ln m\bigr)$ independent samples in total.

\begin{theorem}
\label{thm:kcluster}
There is an FPRAS for estimating the expected length of
$k$-th expensive edge in the minimum spanning tree in the existential uncertainty model.
\end{theorem}

\vspace{-0.4cm}
\section{Minimum Spanning Trees}
\label{sec:mst}
We consider the problem of estimating the expected size of minimum spanning tree in the
locational uncertainty model.
In this section, we briefly sketch how to solve it using our \core\ method.
Recall that the term nodes refers to the vertices $\V$ of the spanning tree
and points describes the locations in $\P$.
For ease of exposition, we assume that for each point, there is only one node that may realize
at this point.

Recall that we use the notation $v \realize s$ to denote the event that
node $v$ is present at point $s$.
Let $\p_{vs}=\Prob[v \realize s]$. Since node $v$ is realized with certainty, we have $\sum_{s\in \P} \p_{vs}=1$. For each point $s\in \P$, we let $p(s)$ denote
the probability that point $s$ is present.
For a set $H$ of points, let $p(H)=\sum_{s\in H} p(s)$, i.e., the expected number of points
present in $H$.
For a set $H$ of points and a set $S$ of nodes, we use $H\A{S}$ to denote the event that
all and only nodes in $S$ are realized to some points in $H$.
If $S$ only contains one node, say $v$, we use the notation $H\A{v}$ as the shorthand for $H\A{\{v\}}$.
Let $H\A{i}$ denote the event $\bigvee_{S:|S|=i}H\A{S}$, i.e., the event that exactly $i$ nodes are in $H$.
We use $\diam(H)$, called the diameter of $H$, to denote $\max_{s,t\in H}\dist(s,t)$.
Let $\dist(p,H)$ be the closest distance between point $p$ and any point in $H$.

\vspace{0.3cm}
\topic{Finding \core}
Firstly, we find in poly-time the \core\  $\H$ as follows:
\vspace{-0.2cm}
\linesnotnumbered
\begin{algorithm}[h]
\caption{Constructing \core\ $\H$ for Estimating $\Exp[MST]$}
\label{alg:estMST}
\nl Among all points $r$ with $p(r)\geq \frac{\epsilon}{16m}$, find the furthest two points $s$ and $t$.\\
\nl Set $\H\leftarrow\B(s, \dist(s,t))=\{s'\in \calP\mid \dist(s',s)\leq \dist(s,t) \}$.\\
\end{algorithm}

\noindent
\begin{lemma}
\label{lm:home}
Algorithm~\ref{alg:estMST} finds a \core\ $\H$ such that
\begin{enumerate}
\item[Q1.] $p(\H)\geq n-\frac{\epsilon}{16}=n-O(\epsilon)$
\item[Q2.] $\Exp[\,\MST\mid \H\A{n}\,]=\Omega\Bigl(\diam(\H)\frac{\epsilon^2}{m^2} \Bigr)$.
\end{enumerate}
Furthermore, the algorithm runs in linear time.
\end{lemma}
\begin{proof}
For each point $r$ that is not in $\H$, we know $p(r)<\frac{\epsilon}{16m}$.
Therefore, we have that
and  $p(\calP\setminus \H)<\frac{\epsilon}{16}$.
and  $p(\H)\geq n-\frac{\epsilon}{16}$.
Consider two cases:
\begin{enumerate}
\item Points $s$ and $t$ relate to different nodes.
In this case, we have that
$$
\Exp[\MST\mid \H\A{n}]\geq \dist(s,t) \Prob[\exists (v,u), v\ne u, v\realize s, u\realize t]=\dist(s,t)p(s)p(t)
\geq \dist(s,t) \frac{\epsilon^2}{256m^2}.
$$
\item Points $s$ and $t$ relate to the same node $v$.
In this case, conditioning on the event that a different node $u$ is realized to an arbitrary point $q$,
$
\Exp[\MST\mid \H\A{n}]\geq \dist(s,q) \Prob[v\realize s]+\dist(t,q) \Prob[v\realize t]
\geq \dist(s,t) \frac{\epsilon}{16m}.
$
\end{enumerate}
In either case, $\H$ satisfies both Q1 and Q2.
\qed
\end{proof}

\topic{Estimating $\Exp[\MST]$}
Let $\calF=\calP\setminus \H$. We rewrite $\Exp[\MST]$ by $\sum_{i\geq 0}\Exp[\MST\mid \calF\A{i}]\cdot \Prob[ \calF\A{i}]$.
We only need to estimate
$\Exp[\,\MST\mid \calF\A{0} \,] \cdot \Prob[ \calF\A{0}]$
and
$\Exp[\,\MST\mid \calF\A{1} \,] \cdot \Prob[ \calF\A{1}]$.

\vspace{0.1cm}
\linesnotnumbered
\begin{algorithm}[h]
\caption{Estimating $\Exp[\,\MST\mid \calF\A{0} \,] \cdot \Prob[ \calF\A{0}]$}
\label{alg:estMSTdetail}
\nl Take $N_0=O\bigl(\frac{nm^2}{\e^4}\ln n\bigr)$ random samples. Set $A\leftarrow \emptyset$ at the beginning.\\
\nl For each sample $G_i$, if it satisfies $\calF\A{0}$, $A\leftarrow A\cup \{G_i\}$.\\
\nl $T_0\leftarrow \frac{1}{N_0}\sum_{G_i\in A}\MST(G_i)$.\\
\end{algorithm}
\vspace{-0.5cm}

\begin{lemma}\label{lm:est1MST}
Algorithm~\ref{alg:estMSTdetail} produces a $(1\pm\epsilon)$-estimate for the first term with high probability.
\end{lemma}

\begin{proof}
Based on the event $\calF\A{0}$, the length of $\MST$ is at most $n\diam(\H)$.
Due to (Q2),
we have a poly-bounded random variable and can therefore obtain a $(1\pm\epsilon)$-estimate for $\Exp[\,\MST\mid \H\A{n}\,]$
using the Monte Carlo method with $O\bigl(\frac{nm^2}{\e^4}\ln n\bigr)$ samples satisfying $\H\A{n}$ (by Lemma~\ref{lm:chernoff}).
By the first property of $\H$, with probability close to 1, a sample satisfies $\H\A{n}$.
So, the expected time to obtain an useful sample is bounded by a constant.
Overall, we can obtain a $(1\pm\epsilon)$-estimate of
the first term with using $N_0=O\bigl(\frac{nm^2}{\e^4}\ln n\bigr)$ samples with high probability.
\qed
\end{proof}

\vspace{-0.1cm}
\linesnotnumbered
\begin{algorithm}[h]
\caption{Estimating $\Exp[\,\MST\mid \calF\A{1} \,] \cdot \Prob[ \calF\A{1}]$}
\label{alg:estMSTdetail2}
\nl Set $B\leftarrow \{s\mid s\in \calF, \dist(s,\H)<\frac{n}{\e } \cdot \diam(\H)\}$.
Let $\Cl(v)$ be the event that $v$ is the only node that is realized to some point $s\in B$.\\
\nl Conditioning on $\Cl(v)$,
take $N_1=O\bigl( \frac{nm^2}{\e^{5}}\ln n\bigr)$ independent samples. \\
Let $A_v\leftarrow \{G_{v,i}\mid 1\leq i\leq N_1\}$ be the set of $N_1$ samples for $\Cl(v)$.\\
\nl $T_v\leftarrow \frac{1}{N_1}\sum_{G_{v,i}\in A_v}\MST(G_{v,i})$  ~~~~(estimating $\Exp[\,\MST\mid \Cl(v)]$)\\
\nl $T_1\leftarrow \sum_{v\in \V}\Bigl(\Prob[\Cl(v)]T_v+\sum_{s\in \calF\setminus B}\Prob[\calF\A{v}\wedge v\realize s]\,\dist(s,\H)\,\Bigr)$.\\
\end{algorithm}

\begin{lemma}
\label{lm:est2MST}
Algorithm~\ref{alg:estMSTdetail2} produces a $(1\pm\epsilon)$-estimate for the second term with high probability.
\end{lemma}

\topic{Analysis}
Note that the number of samples is asymptotically dominated by estimating $\Exp[\,\MST\mid \calF\A{1} \,] \cdot \Prob[ \calF\A{1}]$. For each node $v\in \V$, we take $N_1$ independent samples. Thus, we need to take $O\bigl( \frac{n^2m^2}{\e^{5}}\ln n\bigr)$ independent samples.
Now, we analyze the performance guarantee of our algorithm.
We need to show that the total contribution from the scenarios
where more than one node are not in the \core\ is very small.
We need some notations first.
Suppose $S$ is the set of nodes realized out of \core\ $\H$.
We use $\calF_S$ to denote the set of all possible realizations of all nodes in $S$ to points in $\calF$
(we can think of each element in $\calF_S$ as an $|S|$-dimensional vector where each coordinate
is indexed by a node in $S$ and its value is a point in $\calF$).
Similarly, we denote the set of realizations of $\bS=V\setminus S$ to points in $\H$ by $\calH_{\bS}$.
For any $F_S\in \calF_S$ and $H_{\bS}\in \calH_{\bS}$, we use $(F_S, H_{\bS})$ to denote the event that
both $F_S$ and $H_{\bS}$ happen
and $\MST(F_S, H_{\bS})$ to denote the length of the minimum spanning tree
under the realization $(F_S, H_{\bS})$.
We need the following combinatorial fact.

\begin{lemma}
\label{lm:mstchange}
Consider a particular realization $(F_S,H_{\bS})$,
where $S$ is the set of nodes realized out of $\H$.
$|S|\geq 2$. Let $d=\dist(v_S,u_S)=\min_{v\in S,u\in \bS}\{\dist(u,v)\}$ where $v_S\in F_S$, $u_S\in H_{\bS}$.
The realization $(F_{S'}, H_{\bS'})$ is obtained from $(F_S,H_{\bS})$
by sending the node $v_S$ to $\H$, where $S'=S\setminus{v_S}$.
Then $\MST(F_S, H_{\bS})\leq 4\MST(F_{S'}, H_{\bS'})$.
\end{lemma}

\begin{proof}
We have
$$
4\MST(F'_{S'},H'_{\bS'})\geq 2\MST(F'_{S'},H'_{\bS'})+2d\geq \MST(F'_{S'},H_{\bS})+2d\geq \MST(F_S,H_{\bS})
$$
The second inequality holds since the length of the minimum spanning tree is at most two times
the length of the minimum Steiner tree (We consider $\MST(F'_{S'},H_{\bS})$ as a Steiner tree
connecting all nodes in $F_{S'}\cup H_{\bS}$).
\qed
\end{proof}

The only remaining part for establishing Theorem~\ref{thm:mst}
is to show the following essential lemma.
\begin{lemma}
\label{lm:mstcharge}
For any $\epsilon>0$,
if $\H$ satisfies the properties in Lemma~\ref{lm:home}, we have that
$$
\sum_{i>1}\Exp[\,\MST\mid \calF\A{i}]\cdot \Prob[\calF\A{i}]\leq \epsilon\cdot\Exp[\,\MST\mid \calF\langle 1\rangle]\cdot \Prob[\calF\A{1}].
$$
\end{lemma}

\begin{proof}
We claim that for any $i>1$,
$
\Exp[\,\MST\mid \calF\A{i+1}]\cdot \Prob[\calF\A{i+1}]\leq \frac{\epsilon}{2}\Exp[\,\MST\mid \calF\langle i\rangle]\cdot \Prob[\calF\A{i}].
$
If the claim is true, then we can show the lemma easily by noticing that, for any $n\geq 2$,
$
\sum_{i> 1}\Exp[\,\MST\mid \calF\A{i}] \Prob[\calF\A{i}]
\leq  \sum_{i=1}^{n-1}\bigl(\frac{\epsilon}{2}\bigr)^i  \Exp[\,\MST\mid \calF\langle 1\rangle] \Prob[\calF\A{1}]
\leq  \epsilon\Exp[\,\MST\mid \calF\langle 1\rangle] \Prob[\calF\A{1}].
$
Now, we prove the claim.
First, we rewrite the LHS as follows:
$$
\Exp[\,\MST\mid \calF\A{i+1}]\cdot \Prob[\calF\A{i+1}]
=\sum_{|S|=i+1} \sum_{F_S\in \calF_S} \sum_{H_{\bS}\in \calH_{\bS}}
\bigl(\,
\Prob[(F_S,H_{\bS})]\cdot
\MST(F_S, H_{\bS})\,\bigr),
$$
Similarly, the RHS can be written as:
$$
\Exp[\,\MST\mid \calF\A{i}]\cdot \Prob[\calF\A{i}]
=\sum_{|S'|=i} \sum_{F_{S'}\in\calF_{S'}} \sum_{H_{\bS'}\in \calH_{\bS'}}
\bigl(\,
\Prob[(F_S,H_{\bS})]\cdot
\MST(F_{S'}, H_{\bar{S}'})\,\bigr).
$$
For each pair $(F_S, H_{\bS})$,
let $C(F_S, H_{\bS})=
\Prob[F_S, H_{\bS}]\cdot \MST(F_S, H_{\bS})$.
Consider each pair $(F_S, H_{\bS})$ with $|S|=i+1$ as a seller
and each pair $(F_{S'}, H_{\bS'})$ with $|S'|=i$ as a buyer.
The seller $(F_S, H_{\bS})$ wants to sell the term $C(F_S, H_{\bS})$
and the buyers want to buy all this term.
The buyer $(F_{S'}, H_{\bS'})$ has a budget of $C(F_{S'}, H_{\bS'})$.
We show that there is a charging scheme such that
each term $C(F_S, H_{\bS})$ is fully paid by the buyers and
each buyer spends at most an $\frac{\e}{2}$ fraction of her budget.
Note that the existence of such a charging scheme suffices to prove the claim.

Suppose we are selling the term $C(F_S, H_{\bS})$.
Consider the following charging scheme.
Suppose $v\in S$ is the node closest to any node in $\bS$.
Let $S'=S\setminus \{v\}$ and $F_{S'}$ be the restriction of $F_S$ to all
coordinates in $S$ except $v$.
We say $(F_{S'}, H_{\bS'})$ is consistent with  $(F_S, H_{\bS})$,
denoted as $(F_{S'}, H_{\bS'})\consistent (F_S, H_{\bS})$,
if $H_{\bS'}$ agrees with  $H_{\bS}$ for all vertices in $\bS$.
and $F_{S'}$ agrees with  $F_{S}$ for all vertices in $S\setminus \{v\}$.
Intuitively, $(F_{S'}, H_{\bS'})$ can be obtained from $(F_S,H_{\bS})$
by sending $v$ to an arbitrary point in $\H$.
Let $$
Z(F_S, H_{\bS})=
\sum_{(F_{S'}, H_{\bS'})\consistent (F_S, H_{\bS})} \Prob[(F_{S'}, H_{\bS'})].
$$
We need the following inequality later:
For any fixed $(F_{S'}, H_{\bS'})$,
$$
\sum_{(F_S, H_{\bS})\consistent (F_{S'}, H_{\bS'})}
\frac{\Prob[F_S, H_{\bS}]}
{Z(F_S, H_{\bS})}
\leq \sum_{v\in \bS'} \frac{\Prob(v\in \calF)}{\Prob(v\in \H)}
\leq  \frac{\epsilon}{8}.
$$
To see the inequality,
for a fixed node $v$, consider the quantity
$$
\sum_{(F_S, H_{\bS})\consistent (F_{S'}, H_{\bS'}), \bS=\bS'\setminus \{v\}}
\frac{\Prob[F_S, H_{\bS}]}{Z(F_S, H_{\bS})}.
$$
A crucial observation here is that the denominators of all terms are in fact the same,
by the definition of $Z$,
which is
$
\sum_{} \Prob[(F'_{S'}, H'_{\bS'})],
$
and the summation is over all $(F'_{S'}, H'_{\bS'})$s
which are the same as $(F_{S'}, H_{\bS'})$ except that the location of $v$ is a different point in $\H$.
The numerator is the summation over all $(F_{S}, H_{\bS})$s
which are the same as $(F_{S'}, H_{\bS'})$ except that the location of $v$ is a different point in $\calF$.
Canceling out the same multiplicative terms from the numerators and the denominator,
we can see it is at most $\frac{\Prob(v\in \calF)}{\Prob(v\in \H)}$.

Now, we specify how to charge each buyer.
For each buyer $(F_{S'},H_{\bS'})\sim (F_S, H_{\bS})$,
we charge her the following amount of money
$$
\frac{\Prob[(F_{S'},H_{\bS'})]\cdot C(F_S, H_{\bS})}
{Z(F_S, H_{\bS})}
$$
We can see that $C(F_S, H_{\bS})$ is fully paid by all buyers
consistent with $(F_S, H_{\bS})$.
It remains to show that each buyer $(F_{S'}, H_{\bS'})$
has been charged at most $\frac{\epsilon}{2}C(F_{S'}, H_{\bS'})$.
By the above charging scheme,
the terms $(F_{S}, H_{\bS})$s in LHS that charge buyer $(F_{S'}, H_{\bS'})$
are consistent with $(F_{S'}, H_{\bS'})$.
Now, we can see that the total amount of money charged to buyer $(F_{S'}, H_{\bS'})$
can be bounded as follows:
\begin{align*}
\sum_{(F_S, H_{\bS})\consistent (F_{S'}, H_{\bS'})}
\frac{\Prob[F_{S'},H_{\bS'}]\cdot C(F_S, H_{\bS})}
{Z(F_S, H_{\bS})}
&\leq
4\MST(F_{S'}, H_{\bS'}) \cdot \sum_{(F_S, H_{\bS})\consistent (F_{S'}, H_{\bS'})}
\frac{\Prob[F_{S'},  H_{\bS'}]\cdot \Prob[(F_S, H_{\bS})]}
{Z(F_S, H_{\bS})}
\\
=&\, 4\MST(F_{S'}, H_{\bS'}) \Prob[F_{S'},H_{\bS'}]\cdot
\sum_{(F_S, H_{\bS})\consistent (F_{S'}, H_{\bS'})}
\frac{\Prob[F_S, H_{\bS}]}
{Z(F_S, H_{\bS})}\\
\leq &\, \frac{\epsilon}{2}\MST(F_{S'}, H_{\bS'}) \Prob[F_{S'}, H_{\bS'}]
\end{align*}
The first inequality follows from Lemma~\ref{lm:mstchange}.
This completes the proof.
\qed
\end{proof}

\begin{theorem}
\label{thm:mst}
There is an FPRAS for estimating the expected length of
the minimum spanning tree in the locational uncertainty model.
\end{theorem}

Finally, we remark that
the problem can be solved by a variety of methods.
The \core\ method presented in this section is not the simplest one,
but may be still helpful for understanding a very similar but somewhat more technical application
of the method to minimum perfect matching (see Section~\ref{sec:mm}).


\section{Minimum Perfect Matchings}
\label{sec:mm}

In this section, we consider the minimum perfect matching ($\MM$) problem.
We use the \core\ method.
The same \core\ construction for $\MST$ can not be directly used here since $\MM$ can be much smaller than $\MST$.
For example, suppose there are only two points. There are even number of nodes residing at each point.
In this case, $\MM$ is $0$. Now, if we change the location of one particular node to the other point,
the value of $\MM$ increase dramatically while the value of $\MST$ stays the same.
In some sense, $\MM$ is more sensitive to the location of nodes, hence requires new \core\ construction.
There are two major differences from the algorithm for $\MST$.
First, the \core\ is composed by several clusters of points, instead of a single ball.
Second, we need a more careful charging argument.
\vspace{0.1cm}

\topic{Finding \core}
First, we show how to find in poly-time the \core\  $\H$.
Initially, $\H$ consists of all singleton points, each being a component by itself.
Then, we gradually grow the ball from each point, and merge two components if they touch.
We stop until certain properties Q1 and Q2 are satisfied. See the Pseudo-code in Algorithm~\ref{alg:coreMM} for details.
For a node $v$ and a set $H$ of points, we let $p_v(H)=\sum_{s\in H}p_{vs}$. We use $\diam(H)$, called the diameter of $H$, to denote $\max_{s,t\in H}\dist(s,t)$.

\linesnotnumbered
\begin{algorithm}
\caption{Constructing \core\ $\H$ for Estimating $\Exp[\MM]$}
\label{alg:coreMM}
\nl Initially, $t\leftarrow 0$ and each point $s\in \calP$ is a component $\H_{\{s\}}=\B(s,t)$ by itself.\\
\nl Gradually increase $t$\;
\quad If two different components $\H_{S_1}$ and $\H_{S_2}$ intersect
 (where $\H_S:=\cup_{s\in S}\B(s,t)$)\;
\quad\quad Merge them into a new component $\H_{S_1\cup S_2}$.\\
\nl Stop increasing $t$ while the first time the following two conditions are satisfied
by components at $t$.
\begin{enumerate}
\item[Q1.] For each node $v$, there is a unique component $\H_j$ such that $p_v(\H_j)\geq 1-O(\frac{\e}{nm^3})$.
We call $\H_j$ the \core\ of node $v$, denoted as $\Home(v)$.
\item[Q2.] For all $j$, $|\{v\in\V \mid \H(v)=\H_j\}|$ is even.
\end{enumerate}
\nl Output the stopping time $T$ and the components $\H_1,\ldots, \H_k$.
\end{algorithm}
\vspace{0.3cm}

We need the following lemma which is useful for bounding $\Exp[\MM]$ from below.
\begin{lemma}
\label{lm:lm2core}
For any two disjoint sets $H_1$ and $H_2$ of points, and any node $v$,
we have
$$\Exp[\MM]\geq \min\{p_v(H_1),p_v(H_2)\}\cdot \dist(H_1,H_2)/m.$$
Here, $\dist(H_1,H_2)=\min_{s\in H_1,t\in H_2}\dist(s,t)$.
\end{lemma}

\begin{proof}
Suppose $s =\arg\max_{s'} \{p_{vs'}\mid s' \in H_1\}$,
and $t =\arg\max_{t'}\{p_{vt'}\mid t' \in H_2\}$.
Obviously, we have $p_{vs} \geq \frac{p_v(\H_1)}{m}$ and $p_{vt} \geq \frac{p_v(\H_2)}{m}$.
So it suffices to show $\Exp[\MM]\geq \min\{p_{vs},p_{vt}\}\cdot \dist(s,t)$.
We first see that
\begin{align*}
\Exp[\MM]&\geq p_{vs} \Exp[\MM\mid v \realize s]+ p_{vt} \Exp[\MM\mid v \realize t] \\
&\geq \min\{p_{vs},p_{vt}\} \Bigl(\Exp[\MM\mid v \realize s]
+\Exp[\MM\mid v \realize t]\Bigr).
\end{align*}
Then it is sufficient to prove that $\Exp[\MM\mid v \realize s]+\Exp[\MM\mid v \realize t]\geq \dist(s,t)$.
Fix a realization of all nodes except $v$. Conditioning on this realization, we consider the following two minimum perfect matchings, one for the case $v \realize s$, (denoted as $\MM_1$) and
the other one for $v \realize t$ (denoted as $\MM_2$).
Consider the symmetric difference
$$\MM_1\oplus \MM_2:=(\MM_1\setminus \MM_2)\cup (\MM_2\setminus \MM_1).$$
We can see that it is a path $(s, p_1, p_2, \ldots,p_k ,t)$, such that $(s,p_1)\in \MM_1$,$(p_1,p_2)\in \MM_2,\ldots,$ $(p_k,t)\in \MM_2$.
So $\MM_1+\MM_2 \geq  \dist(s,t)$ by the triangle inequality.
Therefore, we have $\Exp[\MM\mid v \realize s]+\Exp[\MM\mid v \realize t]\geq \dist(s,t)$.
\qed
\end{proof}

By Q1, Q2 and the above lemma, we can show that the following additional property holds.
\begin{lemma}
\label{lm:estMMhome}
\begin{enumerate}
\item
[Q3.] $\Exp[\MM]=\Omega(\frac{\e D}{nm^5})$
where $D=\max_i\{\diam(\H_i)\}$.
\end{enumerate}
\end{lemma}
\vspace{0.1cm}

\begin{proof}
Note that the stopping time $T$ must exist, because the set of all points satisfies the first two properties.
Now, we show that Q3 also holds.
Firstly, note that $D\leq 2mT$.
Secondly, consider $T'=T-\varepsilon$ for some infinitesimal $\varepsilon>0$.
At time $T'$, consider two situations:
\begin{enumerate}
\item There exists a node $v$, such that
$\forall j, p_v(\H_j)< 1-O(\frac{\e}{nm^3})$.
Then there must exist two components $C_1$ and $C_2$ such that $p_v(C_1) > \Omega(\frac{\e}{nm^3})$
and $p_v(C_2) > \Omega(\frac{\e}{nm^3})$.
Moreover, since $C_1$ and $C_2$ are two distinct components, $\dist(C_1,C_2)\geq 2T'$.
Then, by Lemma~\ref{lm:lm2core}, we have $\Exp[\MM]\geq \Omega(\frac{\e}{nm^4}) \cdot 2T\geq \Omega(\frac{\e D}{nm^5})$.
\item Suppose that Q1 is true but Q2 is still false.
Suppose $\H_j$ is a component which homes odd number of nodes.
Note that with probability at least $(1-\frac{1}{nm^3})^{n}\approx 1$,
each node is realized to a point in its \core.
When this is the case, there is at least one node in $\H_j$ that needs to be matched
with some node outside $\H_j$, which incurs a cost of at least
$2T$. \qed
\end{enumerate}
\end{proof}

\topic{Estimating $\Exp[\MM]$}
Let $\H=\cup_i \H_i$.
We use $\H\A{n}$ to denote the event that for each node $v$, $v \realize \Home(v)$.
We denote the event that there are exactly $i$ nodes which are realized out of their \core s by $\calF\A{i}$.
Again, we only need to estimate two terms:
$\Exp[\MM\mid \calF\A{0}]] \cdot \Prob[\calF\A{0}]$ and $
\Exp[\MM\mid \calF\A{1} ] \cdot \Prob[ \calF\A{1}]$.
Using Properties Q1, Q2 and Q3, we can estimate these terms in polynomial time.
Our final estimation is simply the sum of the first two terms.

\linesnotnumbered
\begin{algorithm}[h]
\label{alg:estMM}
\caption{Estimating $\Exp[\,\MM\mid \calF\A{0} \,] \cdot \Prob[ \calF\A{0}]$}
\nl Take $N_1=O(\frac{n^2m^5}{\e^4}\ln n)$ independent samples. Set $A\leftarrow \emptyset$ at the beginning.\\
\nl For each sample $G_i$, if it satisfies $\H\A{n}$, $A\leftarrow A\cup \{G_i\}$.\\
$T_0\leftarrow \frac{1}{N_1}\sum_{G_i\in A}\MM(G_i)$.\\
\end{algorithm}

\begin{lemma}\label{lm:est0MM}
Algorithm~\ref{alg:estMM} produces a $(1\pm\epsilon)$-estimate for the first term with high probability.
\end{lemma}

\begin{proof}
Note that $\Prob[\H\A{n}]$ is close to $1$ (by union bound) and can be computed exactly.
To estimate $\Exp[\MM\mid \H\A{n}]]$, the algorithm takes the average of $N_1=O(\frac{n^2m^5}{\e^4}\ln n)$ samples. Note that conditioning on $\H\A{n}$, the minimum perfect matching could be at most $nD$.
We distinguish the following two cases.
\begin{enumerate}
\item
$\Exp[\MM\mid \H\A{n}]\geq \frac{\e}{2}\Exp[\MM]=\Omega(\frac{\e^2 D}{nm^5})$. We can get a $(1\pm \epsilon)$-approximation using the Monte Carlo method with $O(\frac{n^2m^5}{\e^4}\ln n)$ samples. Therefore $\MM$ is poly-bounded conditioning on $\H\A{n}$.
\item
$\Exp[\MM\mid \H\A{n}]< \frac{\e }{2}\Exp[\MM]$. Then the probability that the sample average is larger than $\epsilon \Exp[\MM]$ is at most $\poly(\frac{1}{n})$ by Chernoff Bound.
We can thus ignore this part safely. \qed
\end{enumerate}
\end{proof}

\vspace{-0.3cm}
\begin{algorithm}[h]
\label{alg:estMM2}
\caption{Estimating $\Exp[\,\MM\mid \calF\A{1} \,] \cdot \Prob[ \calF\A{1}]$}
\nl For each node $v$, set $B_v\leftarrow \{s\mid s\in \calP\setminus \H(v), \dist(s,\H(v))<\frac{4nD}{\e }$. Let $\Cl(v)$ be the event that $v$ is the only node that is realized to some point $s\in B_v$.\\
\nl Conditioning on $\Cl(v)$, take $N_1=O\bigl( \frac{n^2m^5}{\e^{4}}\ln n\bigr)$ independent samples.
Let $A_v\leftarrow \{G_{v,i}\mid 1\leq i\leq N_2\}$ be the set of $N_1$ samples for $\Cl(v)$.\\
\nl $T_v\leftarrow \frac{1}{N_1}\sum_{G_{v,i}\in A_v}\MM(G_{v,i})$  ~~~(estimating $\Exp[\,\MM\mid \Cl(v)]$)\\
\nl $T_1\leftarrow \sum_{v\in \V}\Bigl(\Prob[\Cl(v)]T_v+\sum_{s\in \calF\setminus B_v}\Prob[\calF\A{v}\wedge v\realize s]\,\dist(s,\H(v))\,\Bigr)$.\\
\end{algorithm}

\begin{lemma}\label{lm:est1MM}
Algorithm~\ref{alg:estMM} produces a $(1\pm\epsilon)$-estimate for the second term with high probability.
\end{lemma}

\topic{Analysis}
Note that the number of samples is asymptotically dominated by estimating $\Exp[\,\MM\mid \calF\A{1} \,] \cdot \Prob[ \calF\A{1}]$. For each node $v\in \V$, we take $N_1$ independent samples. Thus, we need to take $O\bigl( \frac{n^3m^5}{\e^{4}}\ln n\bigr)$ independent samples in total.

We still need to show that for $i>1$, the contribution from event $\calF\A{i}$ is negligible.
Suppose $S$ is the set of nodes that are realized out of their \core s.
We use $\calF_S$ and $\calH_{\bS}$
to denote the set of all realizations of the all nodes in $S$ to points out of their \core s,
and the set of realizations of $\bS=V\setminus S$ to points in their \core s respectively.
We use $\MM(F_S, H_{\bS})$ to denote the length of the minimum perfect matching
under the realization $(F_S, H_{\bS})$, where $F_S\in \calF_S$ and $H_{\bS}\in \calH_{\bS}$.
The following combinatorial fact plays the same role
in the charging argument as Lemma~\ref{lm:mstchange} does in the previous section.
Differing from the MST problem,
we can not achieve a similar bound as the one in Lemma~\ref{lm:mstchange} since
$\MM(F_S, H_{\bS})$ may decrease significantly if we send only one node outside its \core\ back to its \core.
However, we show that in such case, if we send one more node back to its \core,
$\MM(F_S, H_{\bS})$ can still be bounded.

We need the following structural result about minimum perfect matchings,
which is essential for our charging argument.

\begin{lemma}
\label{lm:del1or2}
Fix a realization $(F_S,H_{\bS})$.
We use  $\ell(v)$ to denote $\dist(v, \Home(v))$ for all nodes $v\in S$.
Suppose $v_1\in S$ has the smallest $\ell$ value and $v_2$ has the second smallest $\ell$ value.
Let $S'=S\setminus \{v_1\}$, $S''=S'\setminus \{v_2\}$.
Further let $(F_{S'}, H_{\bS'})$ be a realization obtained from $(F_S, H_{\bS})$ by sending $v_1$ to
a point in its \core\ $\H(v_1)$
and
$(F_{S''}, H_{\bS''})$ be a realization obtained from $(F_{S'}, H_{\bS'})$ by sending $v_2$ to
a point in its \core\ $\H(v_2)$.
Then we have that
$
\MM(F_S,H_{\bS})\leq 2(m+2)\MM(F_{S'},H_{\bS'})+2(m+2)\MM(F_{S''},H_{\bS''})
$
\end{lemma}

\begin{proof}
Let $\d=\min_v \ell(v)$ and $D=\max_i \diam(\H_i)$. Note that $d\geq \frac{D}{m}$ as $d\geq 2T$ and $D\leq 2mT$.
We distinguish the following three cases:
\begin{enumerate}
\item
$\MM(F_S,H_{\bS})\leq \frac{\d}{2}$.
Using a similar argument to the one in Lemma~\ref{lm:lm2core}, we have
$$
\MM(F_{S'},H_{\bS'})+\MM(F_S,H_{\bS})\geq \ell(v)=\d
$$
So, we have  $\MM(F_S,H_{\bS})\leq \MM(F_{S'},H_{\bS'})$ in this case.
\item
$\MM(F_S,H_{\bS})\geq (m+2)\d$.
By the triangle inequality,  we can see that
$$
\MM(F_{S'},H_{\bS'})+(m+1)d\geq \MM(F_{S'},H_{\bS'})+d+D\geq \MM(F_S,H_{\bS})
$$
So, we have $\MM(F_S,H_{\bS})\leq (m+2)\MM(F_{S'},H_{\bS'})$.
\item
$\frac{\d}{2}\leq \MM(F_S,H_{\bS})\leq (m+2)\d$.
\begin{enumerate}
\item
$\MM(F_{S'},H_{\bS'})\geq \frac{\d}{2}$.
We directly have $\MM(F_S,H_{\bS})\leq 2(m+2)\MM(F_{S'},H_{\bS'})$.
\item
$\MM(F_{S'},H_{\bS'})\leq \frac{\d}{2}$.
By Lemma~\ref{lm:lm2core}, we have
$$
\MM(F_{S'},H_{\bS'})+\MM(F_{S''},H_{\bS''})\geq \d
$$
Then we have $\MM(F_S,H_{\bS})\leq 2(m+2)\MM(F_{S''},H_{\bS''})$.
\end{enumerate}
\end{enumerate}
In summary, we prove the lemma. \qed
\end{proof}

The remaining is to establish the following key lemma. The proof is
similar to, but more involved than that of Lemma~\ref{lm:mstcharge}.
\begin{lemma}
\label{lm:matchingcharge}
For any $\epsilon>0$,
if $\H$ satisfies the properties Q1, Q2 in Algorithm~\ref{alg:coreMM}, we have that
$$
\sum_{i> 1}\Exp[\MM\mid \calF\A{i}]\cdot \Prob[\calF\A{i}]\leq \epsilon\cdot \Exp[\MM\mid \calF\A{0}] \cdot \Prob[\calF\A{0}]+\epsilon\cdot \Exp[\MM\mid \calF\langle 1\rangle]\cdot \Prob[\calF\A{1}].
$$
\end{lemma}

\begin{proof}
We claim that for any $i > 1$,
$$
\Exp[\MM\mid \calF\A{i+1}]\cdot \Prob[\calF\A{i+1}]\leq
\frac{\epsilon}{6}\bigl(\,\Exp[\MM\mid \calF\A{i}]\cdot \Prob[\calF\A{i}]+\Exp[\MM\mid \calF\A{i-1}]\cdot \Prob[\calF\A{i-1}]\,\bigr)
$$
If the claim is true, the lemma can be proven easily as follows.
For ease of notation, we use $A(i)$ to denote $\Exp[\MM\mid \calF\A{i}]\cdot \Prob[\calF\A{i}]$.
First, we can see that
\begin{align*}
A(i+2)+A(i+1) \leq \frac{\epsilon}{6} A(i+1) +\frac{2\epsilon}{6}A(i) +\frac{\epsilon}{6} A(i-1)\leq \frac{\epsilon}{2}(A(i)+A(i-1)).
\end{align*}
So if $i$ is odd,
$A(i+2)+A(i+1) \leq (\frac{\epsilon}{2})^{(i+1)/2}(A(1)+A(0)).$
Therefore, $\sum_{i>1} A(i)\leq \frac{\epsilon/2}{1-\epsilon/2} (A(1)+A(0))\leq \e(A(1)+A(0))$.
Now, we prove the claim. Again, we rewrite the LHS as
$$
\Exp[\MM\mid \calF\A{i+1}]\cdot \Prob[\calF\A{i+1}]
=\sum_{|S|=i+1} \sum_{F_S} \sum_{H_{\bS}}
\Bigl(\,
\Prob[ F_S,  H_{\bS}]\cdot
\MM(F_S, H_{\bS})\,\Bigr).
$$
Similarly, we have the RHS to be
$$
\Exp[\MM\mid \calF\A{i}]\cdot \Prob[\calF\A{i}]
=\sum_{|S'|=i} \sum_{F_{S'}} \sum_{H_{\bS'}}
\Bigl(\,
\Prob[ F_{S'},  H_{\bS'}]\cdot
\MM(F_{S'}, H_{\bS'})\,\Bigr) \text{ and }
$$
$$
\Exp[\MM\mid \calF\A{i-1}]\cdot \Prob[\calF\A{i-1}]
=\sum_{|S^{''}|=i-1} \sum_{F_{S^{''}}} \sum_{H_{\bS''}}
\Bigl(\,
\Prob[F_{S^{''}},  H_{\bar{S^{''}}}]\cdot
\MM(F_{S^{''}}, H_{\bar{S^{''}}})\,\Bigr).
$$
Let $C(F_S, H_{\bS})=
\Prob[ F_S,  H_{\bS}]\cdot \MM(F_S, H_{\bS})$.
Consider all $(F_{S'},H_{\bS'})$ with $|S^{'}|=i$
and all $(F_{S''},H_{\bS''})$ with $|S^{''}|=i-1$
as buyers.
The buyers want to buy all terms in LHS.
The budget of buyer $(F_{S'},H_{\bS'})/(F_{S''},H_{\bS''})$ is $C(F_{S'},H_{\bS'})/C(F_{S''},H_{\bS''})$.
We show there is a charging scheme such that
each term $C(F_{S}, H_{\bS})$ is fully paid by the buyers and
each buyer spends at most an $\frac{\e}{6}$ fraction of her budget.

Suppose we are selling the term $C(F_S, H_{\bS})$.
Consider the following charging scheme.
Suppose $v_1\in S$ the node that is realized to point $s_1\in \calP\setminus \H(v_1)$
which is the closest point to its \core\ in $F_S$.
Suppose $v_2\in S$ the node that is realized to point $s_2\in \calP\setminus \H(v_2)$
which is the second closest point to its \core\ in $F_S$.
Let $S'=S\setminus \{v_1\}$, $S''=S'\setminus \{v_2\}$.
If $(F_{S'}, H_{\bS'})$ is obtained from $(F_S, H_{\bS})$ by sending $v_1$ to
a point in its \core\ $\Home(v_1)$,
we say $(F_{S'}, H_{\bS'})$
is consistent with  $(F_{S}, H_{\bS})$,
denoted as $(F_{S'}, H_{\bS'})\consistent (F_{S}, H_{\bS})$.
If $(F_{S''}, H_{\bS''})$ is obtained from $(F_{S'}, H_{\bS'})$ by sending $v_2$ to
a point in its \core\ $\Home(v_2)$,
we say $(F_{S''}, H_{\bS''})$
is consistent with  $(F_{S'}, H_{\bS'})$,
denoted as $(F_{S'}, H_{\bS'})\consistent (F_{S}, H_{\bS})$.
Let
$$
Z(F_{S}, H_{\bS})=
\sum_{(F_{S'}, H_{\bS'})\consistent (F_{S}, H_{\bS})} \Prob[(F_{S'}, H_{\bS'})],
\quad\text { and }$$
$$
Z(F_{S'}, H_{\bS'})=
\sum_{(F_{S''}, H_{\bS''})\consistent (F_{S'}, H_{\bS'})} \Prob[F_{S''}, H_{\bS''}]
$$
Now, we claim that for any fixed $(F_{S''}, H_{\bS''})$,
$$
 \sum_{(F_{S'}, H_{\bS'})\consistent (F_{S''}, H_{\bS''})} \frac{\Prob[F_{S'},  H_{\bS'}]}
{Z(F_{S'}, H_{\bS'})}
\leq
\sum_{v\in \bS''} \frac{\Prob[v\notin \Home(v)]}{\Prob[v\in \Home(v)]}.
$$
The proof of the claim is essentially the same as in
Lemma~\ref{lm:mstcharge}.
We first observe that for a fixed node $v=S'\setminus S''$, the denominators of all terms are in fact the same by the definition of $Z$.
Then, the proof can be completed by canceling out the same multiplicative terms from the numerators and the denominator.

Now, we specify how to charge each buyer.
For each buyer $(F_{S'},H_{\bS'})\sim (F_{S}, H_{\bS})$,
we charge $(F_{S'}, H_{\bS'})$ the following amount of money
$$
2(m+2) \Prob[ F_S,  H_{\bS}] \MM(F_{S'}, H_{\bS'})\cdot
\frac{\Prob[F_{S'},  H_{\bS'}]}
{Z(F_{S}, H_{\bS})},
$$
and we charge each buyer $(F_{S''}, H_{\bS''})$ consistent with $(F_{S'}, H_{\bS'})$ the following amount of money
\begin{align*}
2(m+2) \Prob[ F_S'',  H_{\bS''}] \MM(F_{S''}, H_{\bS''})
\cdot
\frac{\Prob[F_{S}, H_{\bS}]}
{Z(F_{S}, H_{\bS})}\cdot \frac{\Prob[F_{S'},  H_{\bS'}]}
{Z(F_{S'}, H_{\bS'})}.
\end{align*}
In this case, we call $(F_{S''}, H_{\bS''})$ a {\em sub-buyer} of the term $C(F_S, H_{\bS})$.
By Lemma~\ref{lm:del1or2}, we can see that $A(F_S, H_{\bS})$ is fully paid.
To prove the claim,
it suffices to show that each buyer $(F_{S'}, H_{\bS'})$ and each sub-buyer $(F_{S''}, H_{\bS''})$
has been charged at most $\frac{\epsilon}{6}A(F_{S'}, H_{\bS'})$ dollars.
By the above charging scheme,
the terms in LHS that are charged to buyer $(F_{S'}, H_{\bS'})$
are consistent with $(F_{S'}, H_{\bS'})$.
Using the same argument as in Lemma~\ref{lm:mstcharge},
we can show that the spending of $(F_{S'}, H_{\bS'})$ as a buyer
is at most
$$\frac{\epsilon}{nm}\cdot \MM(F_{S'}, H_{\bS'})\cdot \Prob[F_{S'},  H_{\bS'}].$$

For notational convenience, we let
$B= 2(m+2)  \MM(F_{S''}, H_{\bS''})\Prob[ F_{S''},  H_{\bS''}]$.
The spending of $(F_{S''}, H_{\bS''})$ as a sub-buyer can be bounded as follows:
\begin{align*}
&B
\cdot\sum_{(F_{S'}, H_{\bS'})\consistent (F_{S''}, H_{\bS''})}
\sum_{(F_{S}, H_{\bS})\consistent (F_{S'}, H_{\bS'})}
\left(\,
\frac{\Prob[F_{S}, H_{\bS}]}{Z(F_{S}, H_{\bS})}\cdot
\frac{\Prob[F_{S'},  H_{\bS'}]} {Z(F_{S'}, H_{\bS'})}
\,\right)\\
\leq &\,B\cdot
 \sum_{(F_{S'}, H_{\bS'})\consistent (F_{S''}, H_{\bS''})}
 \sum_{(F_{S}, H_{\bS})\consistent (F_{S'}, H_{\bS'})}
  \frac{\Prob[F_{S'},  H_{\bS'}]}
{Z(F_{S'}, H_{\bS'})} \\
\leq & B\cdot
 mn \cdot
 \sum_{(F_{S'}, H_{\bS'})\consistent (F_{S''}, H_{\bS''})} \frac{\Prob[F_{S'},  H_{\bS'}]}
{Z(F_{S'}, H_{\bS'})} \\
\leq &\, B\cdot mn \cdot
\sum_{v\in \bS''} \frac{\Prob[v\notin \Home(v)]}{\Prob[v\in \Home(v)]} \\
\leq &\, \frac{\e}{6}\cdot \MM(F_{S''}, H_{\bS''}) \cdot \Prob[F_{S''},  H_{\bS''}]
\end{align*}
In the first inequality, we use the fact that $\frac{\Prob[F_{S}, H_{\bS}]}{Z(F_{S}, H_{\bS})}\leq 1$.
Note that for each $(F_{S'}, H_{\bar{S}'})$, there are at most $mn$ different $(F_{S}, H_{\bar{S}})$
such that $(F_{S}, H_{\bar{S}})\consistent (F_{S'}, H_{\bar{S}'})$. So we have the second inequality.
This completes the proof of the lemma.
\qed
\end{proof}

\begin{theorem}
\label{thm:mm}
Assuming the locational uncertainty model and that the number of nodes is even,
there is an FPRAS for estimating the expected length of
the minimum perfect matching.
\end{theorem}

\topic{Remark}
We have also tried to use the HPF method for this problem.
The problem can be essentially reduced to the following bins-and-balls problem:
Again each ball is thrown to the bins with nonuniform probabilities and
we want to estimate the probability that
each bin contains even number of balls. To the best of our knowledge, the problem is not studied before.
The structure of the problem is somewhat similar to the permanent problem.
We attempted to use the MCMC technique developed in \cite{jerrum2004permanent},
but the details become overly messy and we have not been able to provide a complete proof.

\section{Minimum Cycle Covers}
\label{app:cc}
In this section, we consider the expected length of minimum cycle cover problem.
In the deterministic version of the cycle cover problem,
we are asked to find a collection of node-disjoint cycles such that
each node is in one cycle and the total length is minimized.
Here we assume that each cycle contains at least two nodes.
If a cycle contains exactly two nodes, the length of the cycle is two times the
distance between these two nodes.
The problem can be solved in polynomial time by
reducing the problem to a minimum bipartite perfect matching problem.
\footnote{
If we require each cycle consist at least three nodes,
the problem is still poly-time solvable by a reduction to minimum perfect matching by Tutte~\cite{tutte54}.
Hartvigsen~\cite{harvigsen84} obtained a  polynomial time algorithm for
minimum cycle cover with each cycle having at least 4 nodes
Cornu\'{e}jols and Pulleyblank \cite{corn1980} have reported that Papadimitriou showed the NP-completeness of
minimum cycle cover with each cycle having at least 6 nodes.
}
W.l.o.g., we assume that no two edges in $\P \times \P$ have the same length.
For ease of exposition, we assume that for each point, there is only one node that may realize
at this point.
In principle, if more than one nodes may realize at the same point, we can
create multiple copies of the point co-located at the same place,
and impose a distinct infinitesimal distance between each pair of copies,
to ensure that no two edges have the same distance.


We need the notion of the nearest neighbor graph, denoted by $\NN$ .
For an undirected graph, an edge $e=(u,v)$ is in the nearest neighbor graph
if $u$ is the nearest neighbor of $v$, or vice versa.
We also use $\NN$ to denote its length.
$\Exp[\NN]$ can be computed exactly in polynomial time \cite{kamousi2011stochastic}.
\eat{
In fact, this can be seen by noting that
$
\Exp[\NN]=\sum_{e}c_e \cdot \Prob[e\in \NN]
$
and
$\Prob[e\in \NN] $ can be computed in poly-time.
To see this, note that $e=(s,t)$ is in $\NN$ if and only if at least one of the following two events happen
(1) both points $s$ and $t$ are chosen and no other points are chosen in $\B(s, \dist(e))$ and
(2) both points $s$ and $t$ are chosen and no other points are chosen in $\B(s, \dist(e))$.
}
As a warmup, we first show that $\Exp[\NN]$ is a 2-approximation of $\Exp[\CC]$ in the following lemma.
\begin{lemma}
\label{lmNNCC}
$
\Exp[\NN]\leq \Exp[\CC] \leq 2\Exp[\NN].
$
\end{lemma}

\begin{proof}
We show that $\NN\leq \CC\leq 2\NN$ satisfies for each possible realization.
We prove the first inequality.
For each node $u$, there are two edges incident on $u$.
Suppose they are $e_{u1}$ and $e_{u2}$. We have
$
\CC=\frac{\sum_{u}(\dist(e_{u1})+\dist(e_{u2}))}{2}\geq \NN.
$
The second inequality can be seen by doubling all edges in $\NN$
and the triangle inequality. \qed
\end{proof}

We denote the longest edge in $\NN$ (and also its length) by  $\L$.
Note that $\L$ is also a random variable.
By the law of total expectation, we estimate $\Exp[\CC]$ based on the following formula:
$$
\Exp[\CC]=\sum_{e\in \P\times\P}\Prob[\L=e]\cdot \Exp[\CC\mid \L=e]
$$
It is obvious to see that $\frac{\NN}{n}\leq \L\leq \NN$.
Combined with Lemma~\ref{lmNNCC}, we have that
\begin{align}
\label{eq:CCbound}
\dist(e)\leq \Exp[\CC\mid \L=e]\leq 2n\dist(e).
\end{align}
However, it is not clear to us how to estimate $\Prob[\L=e]$
and how to take samples conditioning on event $\L=e$ efficiently.
To circumvent the difficulty, we consider some simpler events.
Consider a particular edge $e=(s,t) \in \P\times\P$.
Denote as $N_s(t)$ the event that the nearest neighbor of $s$ is $t$.
Let $L_{st}$ be the event the longest edge $\L$ in $\NN$ is $e=(s,t)$.
Let $A_s(t)=N_s(t)\wedge L_{st}$. First we rewrite $\Exp[\CC\mid \L=e]\cdot \Prob[\L=e]$ by 
\begin{align*}
\Exp[\CC\mid \L=e]\cdot \Prob[\L=e] =& \Exp[\CC\mid A_s(t)\vee A_t(s)]\cdot\Prob[A_s(t)\vee A_t(s)] \\
= & \Exp[\CC\mid A_s(t)]\cdot \Prob[A_s(t)]+ \Exp[\CC\mid A_t(s)]\cdot \Prob[A_t(s)]\\
  &- \Exp[\CC\mid A_s(t)\wedge A_t(s)]\cdot \Prob[A_s(t)\wedge A_t(s)]
\end{align*}
Now, we show how to estimate $\Exp[\CC\mid A_s(t)]\cdot \Prob[A_s(t)]$ for each edge $e=(s,t)$.
The other two terms can be estimated in the same way.
Also notice that the third term is less than both the first term and the second term.
Therefore, for any points $s$ and $t$, we have the following fact which is useful later:
\begin{align}
\label{eq:CCbound2}
\Exp[\CC]\geq \Exp[\CC\mid \L=e]\cdot \Prob[\L=e]\geq \Exp[\CC\mid A_s(t)]\cdot \Prob[A_s(t)].
\end{align}
By the above inequality, we can see that the total error for estimating the three terms is negligible compared to $\Exp[\CC\mid \L=e]\cdot \Prob[\L=e]$. Moreover, we have that
\begin{align*}
\Exp[\CC\mid A_s(t)]\cdot\Prob[A_s(t)]=\Exp[\CC\mid A_s(t)]\cdot\Prob[L_{st}\wedge N_s(t)]
\\=\Exp[\CC\mid A_s(t)]\cdot \Prob[L_{st}\mid N_s(t)]\cdot \Prob[N_s(t)]
\end{align*}
Suppose $v$ is the node that may be realized to point $s$
and $u$ is the node that may be realized to point $t$.
We use $\B$ as a shorthand notation for $\B(s,\dist(s,t))$.
We first observe that $\Prob[N_s(t)]$ can be computed exactly in poly-time as follows:
$$
\Prob[N_s(t)]=p_{vs}\cdot p_{ut}\cdot \prod_{w\ne v,u} \bigl(1-p_w(\B) \bigr)
$$
Also note that we can take samples conditioning on the event $N_s(t)$
(the corresponding probability distribution for node $v$ is:
$\Prob[v\realize r \mid N_s(t)] = \frac{p_{vr}}{1-p_w(\B)}$).

\vspace{0.2cm}
\noindent
\topic{ Estimating $\Exp[\,\CC\mid A_s(t)\,]\cdot \Prob[L_{st}\mid N_s(t)]$}
Next, we show how to estimate $\Exp[\CC\mid A_s(t)]\cdot\Prob[L_{st}\mid N_s(t)]$.
The high level idea is the following.
We take samples conditioning on $N_s(t)$.
If $\Prob[L_{st}\mid N_s(t)]$ is large (i.e., at least $1/\poly(nm)$),
we can get enough samples satisfying $L_{st}$, thus $A_s(t)$.
Therefore, we can get $(1\pm\e)$-approximation for
both $\Prob[L_{st}\mid N_s(t)]$ and $\Exp[\CC\mid A_s(t)]$ in poly-time
(we also use the fact that if $A_s(t)$ is true, $\CC$ is at least $\dist(s,t)$ and at most $2n\dist(s,t)$).
However, if $\Prob[L_{st}\mid N_s(t)]$ is small, it is not clear how to obtain a reasonable
estimate of this value.
In this case, we show the contribution of the term to our final answer is extremely small
and even an inaccurate estimation of the term will not affect our answer
in any significant way with high probability.

Now, we elaborate the details.
We iterate the following steps for $N$ times ($N=O(\frac{n^2m^4}{\e^3}(\ln n+\ln m))$ suffices). Since there are $O\bigl(m^2\bigr)$ different edges between points, we totally need $O(\frac{n^2m^6}{\e^3}(\ln n+\ln m))$ iterations.
\begin{itemize}
\item Suppose we are in the $i$th iteration.
We take a sample $G_i$ of the stochastic graph conditioning on the event $N_{s}(t)$.
We compute the nearest neighbor graph $\NN(G_i)$ and
the minimum length cycle cover $\CC(G_i)$.
If $e=(s,t)$ is the longest edge in $\NN(G_i)$, let $I_i=1$.
Otherwise $I_i=0$.
\end{itemize}
Our estimate of $\Exp[\,\CC\mid A_s(t)\,]\cdot \Prob[L_{st}\mid N_s(t)]$
is the following:
$$
\left( \frac{\sum_{i=1}^{N} I_i \cdot\CC(G_i) }{\sum_{i=1}^N I_i} \right)\left( \frac{\sum_{i=1}^N I_i}{N} \right)
=
\frac{\sum_{i=1}^{N} I_i\cdot \CC(G_i) }{N}
$$
It is not hard to see that the expectation of $\frac{\sum_{i=1}^{N} I_i\cdot \CC(G_i) }{N}$ is exactly
$\Exp[\,\CC\mid A_s(t)\,]\cdot \Prob[L_{st}\mid N_s(t)]$.

We distinguish the following two cases:
\begin{enumerate}
\item
$ \Prob[L_{st}\mid N_s(t)]\geq \frac{\epsilon}{2nm^4}$.
By Lemma~\ref{lm:chernoff}, $\frac{\sum_{i=1}^N I_i}{N}\in (1\pm \epsilon)\Prob[L_{st}\mid N_s(t)]$ with high probability.
In this case, we have enough successful samples (samples with $I_i=1$) to guarantee that
$\frac{\sum_{i=1}^{N} I_i \CC(G_i) }{\sum_{i=1}^N I_i}$
is a ($1\pm \e$)-approximation of $\Exp[\,\CC\mid A_s(t)\,]$ with high probability,
again by Lemma~\ref{lm:chernoff}.
We note that under the condition $A_s(t)$, we can get a $(1\pm \e)$-approximation
since $\CC$ is at least $\dist(s,t)$ and at most $2n\dist(s,t)$.
\item
$\Prob[L_{st}\mid N_s(t)]< \frac{\epsilon}{2nm^4}$.
We note that $I_i=0$ means that while $N_s(t)$ happens, the longest edge $\L$ in $\NN$ is longer than $e=(s,t)$.
Suppose $e'=(s',t')$ is the edge with the maximum $\Prob[L_{s't'}|N_s(t)]$.
Since $\Prob[L_{st}\mid N_s(t)]\leq \frac{\epsilon}{2nm^4}$, $e'=(s',t')$ must be different from $e=(s,t)$
and $\Prob[L_{s't'}\mid N_s(t)]\geq \frac{4nm^2}{\e}\Prob[L_{st}\mid N_s(t)]$.
Hence, we have that
\begin{align*}
\Exp[\CC\mid A_s(t)]\cdot\Prob[A_s(t)]&=
\Exp[\CC\mid A_s(t)]\cdot \Prob[L_{st}\mid N_s(t)]\cdot \Prob[N_s(t)] \\
&\leq 2n \cdot \dist(s,t) \cdot \frac{\e}{4nm^2}\cdot \Prob[L_{s't'}\mid N_s(t)]\cdot \Prob[N_s(t)] \\
&\leq \frac{\e}{2m^2}\cdot \dist(s',t') \cdot \Prob[L_{s't'}\mid N_s(t)]\cdot \Prob[N_s(t)] \\
&\leq \frac{\e}{2m^2}\cdot \Exp[\,\CC\mid A_{s'}(t')] \cdot \Prob[L_{s't'}] \\
&\leq \frac{\e}{2m^2}\cdot \Exp[\CC]
\end{align*}
The first and third inequalities are due to \eqref{eq:CCbound} and the fourth are due to \eqref{eq:CCbound2}.
By Chernoff Bound, we have that
$$
\Prob\left[\frac{\sum_{i=1}^{N} I_i\cdot \CC(G_i) }{N}\geq \frac{\e}{m^2}\cdot \Exp[\CC]\right]\leq \frac{e^{-n}}{m^2}
$$
Then, with probability at least $1-\poly(\frac{1}{n})$,
the contribution from all such edges is less than $\e\Exp[\CC]$.
\end{enumerate}

Summing up, we have obtained the following theorem.
\begin{theorem}
\label{thm:cc}
There is an FPRAS for estimating the expected length of
the minimum length cycle cover in both the locational uncertainty model
and the existential uncertainty model.
\end{theorem}



Finally, we remark that our algorithm also works in presence of both locational uncertainty and node uncertainty, i.e.,
the existence of each node is a Bernoulli random variable.
It is not hard to extend our technique to handle the case where each cycle is required to contain at least three nodes.
This is done by considering the longest edge in the $2\NN$ graph (each node connects to the nearest and the second nearest
neighbors).
The extension is fairly straightforward and we omit the details here.


\section{$k$th Longest $m$-Nearest Neighbor}
\label{app:kmNN}

We consider the problem of computing the expected length of the $k$th longest $m$-nearest neighbor
(i.e., for each point, find the distance to its $m$-nearest neighbor, then compute the $k$th longest one among these distances)
in the existential uncertainty model.
We use $\KMNN$ to denote the length of the $k$th longest $m$-nearest neighbor.

Similar to $k$-clustering, we use the HPF $\hpf$ for estimating $\Exp[\KMNN]$. We call a component a small component if it contains at most $m$ present points.
Let the random variable $Y$ be the largest integer $i$ such that there are at most $k-1$ present points
among those small components in $\Gamma_i$. We can see that if $Y=i$ then the special component $\nu_i$ is not a small component, while both $\mu'_{i+1}$ and $\mu''_{i+1}$ should not be empty, and one of $\mu'_{i+1}$ and $\mu''_{i+1}$ must be a small component. Moreover, $\Gamma'_i$ contains at most $k-1$ present points among those small components.

We can rewrite $\Exp[\KMNN]$ by $\Exp[\KMNN]=\sum_{i=1}^{m} \Prob[Y=i]\Exp[\KMNN\mid Y=i]$.
By the Property P1 and P2 of $\hpf$, we directly have the following lemma.

\begin{lemma}
\label{lm:kmNN}
Conditioning on $Y=i$, it holds that $\dist(e_{i})\leq \KMNN\leq m \dist(e_{i})$.
\end{lemma}

For a partition $\Gamma$ on $\P$, we use $\Gamma\A{\#j,\leq m}$ to denote the event that there are exactly $j$ present points among those small components in $\Gamma$. The remaining task is to show how to compute $\Prob[Y=i]$ and how to estimate $\Exp[\KMNN\mid Y=i]$. We first prove the following lemma.

\begin{lemma}
\label{lm:samJM}
For a partition $\Gamma$ on $\P$, we can compute $\Prob[\Gamma\A{\#j,\leq m}]$ in polynomial time. Moreover, there exists a polynomial time
sampler for sampling present points in $\Gamma$ conditioning on $\Gamma\A{\#j,\leq m}$.
\end{lemma}

\begin{proof}
W.l.o.g, we assume that the components in $\Gamma$ are $C_1,\ldots, C_{n}$. We denote $E[a,b]$ the event that among the first $a$ components, exactly $b$ points are present in those small components. We denote the probability of $E[a,b]$ by $\Prob[a,b]$. Note that our goal is to compute $\Prob[n,j]$. We have the following dynamic program:

\begin{enumerate}
\item If $\sum_{1\leq l\leq a}\min\{m,|C_l|\} < b$, $\Prob[a,b]=0$. If $b=0$, $\Prob[a,b]=\prod_{1\leq l\leq a}(\Prob[C_l\A{0}]+\Prob[C_l\A{\geq m+1}])$.
\item For $1\leq b\leq \sum_{1\leq l\leq a}\min\{m,|C_l|\}$, $\Prob[a,b]=\sum_{0\leq l\leq m }\Prob[C_a\A{l}]\cdot \Prob[a-1,b-l]+\Prob[C_a\A{\geq m+1}]\cdot \Prob[a-1,b]$.
\end{enumerate}

Thus we can compute $\Prob[n,j]$ in polynomial time. Similar to Lemma~\ref{lm:samTj}, we can also construct a polynomial uniform sampler.
\qed
\end{proof}

To prove Theorem~\ref{thm:kmNN}, we only need the following lemma.

\begin{lemma}
\label{lm:samkmNN}
We can compute $\Prob[Y=i]$ in polynomial time. Moreover, there exists a polynomial time sampler conditioning on $Y=i$.
\end{lemma}

\begin{proof}
By the definition of $Y=i$, we can rewrite $\Prob[Y=i]$ as follows:
\begin{align*}
\Prob[Y=i]&=\sum_{1\leq n_1\leq m, m+1-n_1\leq n_2\leq m} \Prob[\mu'_{i+1}\A{n_1}]\cdot \Prob[\mu''_{i+1}\A{n_2}]\cdot
\left( \sum_{k-n_1-n_2\leq l\leq k-1} \Prob[\Gamma'_i\A{\#l,\leq m}]\right) \\
&+\sum_{m+1\leq n_1\leq |\mu'_{i+1}|, 1\leq n_2\leq m} \Prob[\mu'_{i+1}\A{n_1}]\cdot \Prob[\mu''_{i+1}\A{n_2}]\cdot
\left( \sum_{k-n_2\leq l\leq k-1} \Prob[\Gamma'_i\A{\#l,\leq m}]\right) \\
&+\sum_{1\leq n_1\leq m, m+1\leq n_2\leq |\mu''_{i+1}|} \Prob[\mu'_{i+1}\A{n_1}]\cdot \Prob[\mu''_{i+1}\A{n_2}]\cdot
\left( \sum_{k-n_1\leq l\leq k-1} \Prob[\Gamma'_i\A{\#l,\leq m}]\right)
\end{align*}
Note that we can compute $\Prob[Y=i]$ in polynomial time by Lemma~\ref{lm:samJM}. Using the same argument as in Lemma~\ref{lm:samCP_k}, we can construct a polynomial uniform sampler conditioning on $Y=i$. By Lemma~\ref{lm:kmNN}, we only need to take $O(\frac{m}{\e^2}\ln m)$ independent samples for estimating $\Exp[\KMNN\mid Y=i]$. So we take $O(\frac{m^2}{\e^2}\ln m)$ independent samples in total.
\qed
\end{proof}

\begin{theorem}
\label{thm:kmNN}
There is an FPRAS for estimating the expected length of
the kth longest m-nearest neighbor in the existential uncertainty model.
\end{theorem}

\section{Conclusion}
%
%

Our work leaves a number of interesting open problems.
One interesting open problem is to estimate
the expected value
of the minimum cost matching of a certain cardinality (instead of the perfect matching).
It is not clear how to extend our technique to handle this problem.
Moreover, computing the threshold probabilities $\Pr[\mathsf{Obj}\leq 1]$ and $\Pr[\mathsf{Obj}\geq 1]$
for most problems, except closest pair and diameter, have not been studied yet.
The only hardness result we are aware of is that computing $\Pr[\MST\leq 1]$
is \#P-hard to approximate to any factor \cite{kamousi2011stochastic}.



\section{Acknowledgements}
Part of this work was done while JL visited
the Simons Institute for the Theory of Computing.
We would like to thank Alistair Sinclair, Jeff Phillips, Pinyan Lu, Yitong Yin, Uri Zwick for helpful discussions.

\bibliographystyle{plain}
\bibliography{smst}

\appendix

\section{Missing Proofs}

\subsection{Closest Pair}

\textbf{Lemma~\ref{lm:cp}} \emph{Steps 1,2,3 in Algorithm~\ref{algo:cp} provide
	$(1\pm \epsilon)$-approximations for $\Prob[\calF\A{i} \wedge \CP\leq 1]$
	for $i=0,1,2$ respectively,
	with high probability.}

\vspace{0.3cm}
\begin{proof}
As we just argued,  $\Prob[\calF\A{1}\wedge \CP \leq 1]$ can be estimated since
$I(\CP\leq 1)$, conditioned on $\calF\A{0}$, is poly-bounded.
For estimating $\Prob[\calF\A{1}\wedge \CP \leq 1]$, we first rewrite this term by $\sum_{s_i\in \calF}\Prob[\calF\A{\{s_i\}}\wedge \CP \leq 1]$.
For a point $s_i\in \calF$, note that $\Prob[\calF\A{\{s_i\}}\wedge \CP \leq 1]=\Prob[\calF\A{\{s_i\}}]\cdot \Prob[\CP \leq 1\mid \calF\A{\{s_i\}}]$. Since we have that $p_i(1-\frac{\e}{m})\leq \Prob[\calF\A{\{s_i\}}]\leq p_i$ by the first property of the \core\ $\calH$, we can use $p_i$ to estimate $\Prob[\calF\A{\{s_i\}}]$. For estimating $\Prob[\CP \leq 1\mid \calF\A{\{s_i\}}]$,
we denote $\B_{s_i}=\{t\in \calH:\dist(s_i,t)\leq 1\}$.
If $\B_{s_i}$ is not empty,
we can use Monte Carlo for estimating
$\Prob[\CP \leq 1\mid \calF\A{\{s_i\}}]$ since its value is at least $\frac{\epsilon}{m^2}$.
Otherwise,
computing $\Prob[\CP \leq 1\mid \calF\A{\{s_i\}}]$
is equivalent to computing
$\Prob[\CP \leq 1\mid \calF\A{0}]$
in the instance without $s_i$ (since $s_i$ is at distance more than 1 from any other point).
The proof for $\Prob[\calF\A{2}\wedge \CP \leq 1]$ is almost the same and we do not repeat it.
\qed
\end{proof}

\subsection{Minimum Spanning Tree}

\textbf{Lemma~\ref{lm:est2MST}}  \emph{
Algorithm~\ref{alg:estMSTdetail2} produces a $(1\pm\epsilon)$-estimate for the second term with high probability.}

\vspace{0.3cm}
\begin{proof}
To compute the second term, we first rewrite it as follows:
\begin{align*}
\Exp[\,\MST\mid \calF\A{1} \,]\cdot \Prob[\calF\A{1}]
=\sum_{v\in \V} \Bigl(\,\sum_{s\in F} \Prob[\calF\A{v}\wedge v\realize s]\, \Exp[\,\MST\mid \calF\A{v}, v\realize s ]\,\Bigr)
\end{align*}
Fix a  node $v$.
To estimate $\sum_{s\in F} \Prob[\calF\A{v}\wedge v\realize s]\, \Exp[\MST\mid \calF\A{v}, v\realize s ]$, we consider the following two situations:
\begin{enumerate}
\item
Point $s\in B$, i,e, $\dist(s,\H)<\frac{n}{\e } \cdot \diam(\H)$.

We estimate the sum for all $s\in B$. Notice that the sum is in fact $\Prob[\Cl(v)]\cdot\E[\,\MST\mid \Cl(v) ]$.
We can see that $\Prob[\Cl(v)]$ can be computed exactly in linear time.
We argue that the quality of the estimation taken on $N_1=O\bigl( \frac{nm^2}{\e^{5}}\ln n\bigr)$ samples is sufficient
by considering the following two cases:
    \begin{enumerate}
    \item Assume that $
    \E[\,\MST\mid \Cl(v) ]
    \geq \frac{1}{2}\Exp[\,\MST \mid \H\A{n}]\geq \Omega\Bigl(\frac{\e^2}{m^2}\Bigr)\diam(\H).
    $
    In this case, we have a poly-bounded random variable.
    This is because under the condition $\Cl(v)$,
    the maximum possible length of any minimum spanning tree is $O(\frac{n}{\e}\diam(\H))$.
    Hence we can use Monte Carlo to get a $(1\pm \epsilon)$-approximation
    of $\E[\,\MST\mid \Cl(v) ]$ with $O\bigl( \frac{nm^2}{\e^{5}}\ln n\bigr)$ samples.
    \item Otherwise, we assume that
    $
    \E[\,\MST\mid \Cl(v) ]
    \leq \frac{1}{2}\Exp[\,\MST \mid \H\A{n}]].
    $
    Let $V_0$ be the collection of these nodes. The probability that the sample average is larger than $\Exp[\MST\mid \H\A{n}]] $
    is at most $\poly(\frac{1}{n})$ by Chernoff Bound. The probability that for all nodes $v\in V_0$, the sample average are
    at most $\Exp[\MST\mid \H\A{n}]]$ is at least $1-\poly(\frac{1}{n})$ by union bound.
    If this is the case, we can see their total contribution
    to the final estimation of $\E[\MST]$ is less than $\e \Exp[\,\MST \mid \H\A{n}]] \Prob[\H\A{n}] $.
    In fact, this is because
    \begin{align*}
     \sum_{v\in V_0} \Prob[\Cl(v)]\cdot T_v& \leq \sum_{v\in V_0}\Prob[\Cl(v)]\cdot \Exp[\,\MST \mid \H\A{n}]]
     <\e \Exp[\,\MST \mid \H\A{n}]] \Prob[\H\A{n}].
    \end{align*}
    The second inequality is due to the fact that $\sum_{v\in V_0}\Prob[\Cl(v)]\leq n-p(\H)<\e/16<\e\Prob[\H\A{n}]$.
    \end{enumerate}
\item
Point $s\in \calF\setminus B$, each term has $\dist(s,\H)>\frac{n}{\e } \cdot \diam(\H)$. \\
We just use $\dist(s,\H)$
as the estimation of $\Exp[\MST\mid \calF\A{v}, v\realize s ]$.
This is because the length of $\MST$ is always at least $\dist(s,\H)$
and at most $\dist(s,\H)+n\cdot\diam(\H) \leq (1+\epsilon) \dist(s,\H)$.
\qed
\end{enumerate}
\end{proof}

\subsection{Minimum Perfect Matching}

\textbf{Lemma~\ref{lm:est1MM}} \emph{
Algorithm~\ref{alg:estMM} produces a $(1\pm\epsilon)$-estimate for the second term with high probability.}

\vspace{0.3cm}
\begin{proof}
To compute the second term, we first rewrite it as follows:
$$\Exp[\MM\mid \calF\A{1} ]\cdot \Prob[\H\A{1}]
=\sum_{v\in \V} \Bigl(\,\sum_{s\notin \Home(v)} \Prob[\calF\A{v}\wedge v\realize s]\
\Exp[\MM\mid \calF\A{v}, v\realize s ]\,\Bigr).
$$
Fix a particular node $v$.
To estimate $\sum_{s\in \calF} \Prob[\calF\A{v}\wedge v\realize s]\, \Exp[\MM\mid \calF\A{v}, v\realize s ]$, we consider the following two situations:
\begin{enumerate}
\item
Point $s\in B_v$, i,e, $\dist(s,\H(v))<\frac{4nD}{\e }$.

We estimate the sum for all $s\in B^v$. Notice that the sum is
in fact $\Prob[\Cl(v)]\cdot\E[\,\MM\mid \Cl(v) ]$.
We can see that $\Prob[\Cl(v)]$ can be computed exactly in linear time.
We argue that the quality of the estimation taken on $N_2=O\bigl( \frac{n^2m^5}{\e^{4}}\ln n\bigr)$ samples is poly-bounded
by considering the following two cases:
    \begin{enumerate}
    \item Assume that $
    \E[\,\MM\mid \Cl(v) ]
    \geq \frac{1}{2}\Exp[\,\MM \mid \H\A{n}]= \Omega\Bigl(\frac{\e D}{nm^5}\Bigr).
    $
    In this case, our estimation is poly-bounded.
    This is because under the condition $\Cl(v)$,
    the maximum possible length of any minimum perfect matching is $O(\frac{nD}{\e})$.
    Hence we can use Monte Carlo to get a $(1\pm \epsilon)$-approximation
    of $\E[\,\MM\mid \Cl(v) ]$ with $O\bigl( \frac{n^2m^5}{\e^{4}}\ln n\bigr)$ samples.
    \item Otherwise, we assume that
    $
    \E[\,\MM\mid \Cl(v) ]
    \leq \frac{1}{2}\Exp[\,\MM \mid \H\A{n}]].
    $
    Let $V_0$ be the collection of these nodes. The probability that the sample average is larger than $\Exp[\MM\mid \H\A{n}]] $
    is at most $\poly(\frac{1}{n})$ by Chernoff Bound. The probability that for each node $v\in V_0$, the sample average is
    at most $\Exp[\MM\mid \H\A{n}]]$ is at least $1-\poly(\frac{1}{n})$ by union bound.
    If this is the case, we can see their total contribution
    to the final estimation of $\E[\MM]$ is less than $\e \Exp[\,\MM \mid \H\A{n}]] \Prob[\H\A{n}] $.
    In fact, this is because
    $$
    \sum_{v\in V_0} \Prob[\Cl(v)]\cdot T_v\leq \sum_{v\in V_0}\Prob[\Cl(v)]\cdot \Exp[\,\MM \mid \H\A{n}]]
     <\e \Exp[\,\MM \mid \H\A{n}]] \Prob[\H\A{n}].
    $$
    The second inequality is due to the fact that $\sum_{v\in V}\Prob[\Cl(v)]\leq n-\sum_{v\in V_0} p_v(\H(v))\leq \frac{\e}{m^3}<\e \Prob[\H\A{n}]$.
    \end{enumerate}
\item
Point $s\in \calP\setminus (B_v\cup \H(v))$, each term has $\dist(s,\H(v))>\frac{4nD}{\e }$.
The algorithm uses $\dist(s,\H(v))$
as the estimation of $\Exp[\MM\mid \calF\A{v}, v\realize s ]$.
Note that the length of $\MM$ is always at least $\dist(s,\H(v))-nD\geq (1-\frac{\e}{4})\dist(s,\H(v))$. This is because such an instance $\MM$ contains a path from $s$ to some point $t\in \H(v)$ deleting no more than $n$ segments of length at most $D$ (each segment is in some $\H_j$).
On the other hand, the length of $\MM$ is at most $\dist(s,\H(v))+nD \leq (1+\frac{\e}{4}) \dist(s,\H(v))$. So it is a $(1\pm\epsilon)$-estimation.
\qed
\end{enumerate}
\end{proof}

\section{The Closest Pair Problem}

\subsection{Estimating $k$th Closest Pair in the Existential Uncertainty Model}
\label{app:CP}

Again, we construct the HPF $\hpf$.
Let the random variable $Y$ be the largest integer $i$ such that there are at least $k$ point collisions in $\Gamma_i$.
Here we use a point collision to denote that a pair of points are present in the same component.
Note that if there are exactly $i$ points in a component, the amount of point collisions in this component is ${i\choose 2}$.
We denote as $\Gamma\A{\#j}$ the event that there are exactly $j$ point collisions among the partition $\Gamma$ on $\P$. Similarly, we can rewrite $\Exp[\KCP]$ by $\Exp[\KCP]=\sum_{i=1}^{m-1} \Prob[Y=i]\Exp[\KCP\mid Y=i]$.

We use dynamic programming technique to achieve an $\FA$ for computing $\Exp[\KCP]$.
Note that conditioning on $Y=i$, the value of $\KCP$ is between $\dist(e_{i})$ and $m\cdot \dist(e_{i})$. So we only need to show the following lemma.

\begin{lemma}
\label{lm:samCP_k}
We can compute $\Prob[Y=i]$ in polynomial time. Moreover, there exists a polynomial time sampler conditioning on $Y=i$.
\end{lemma}

\begin{proof}
We denote $E[a,b]$ ($1\leq a\leq i-1$, $b\leq k$) the event that among the first $a$ components in $\Gamma'_i $, there are exactly $b\leq k$ point collisions. We denote the probability of $E[a,b]$ by $\Prob[a,b]$. We give the dynamic programming as follows.

\begin{enumerate}
\item If $\sum_{1\leq j\leq a}{|C_{j}|\choose 2} < b$, $\Prob[a,b]=0$. If $b = 0$, $\Prob[a,b]=\prod_{1\leq j\leq a}\Prob[C_j\A{\leq 1}]$. If $b<0$, $\Prob[a,b]=0$.
\item If $\sum_{1\leq j\leq a}{|C_{j}|\choose 2} \geq b$, $1\leq b\leq k$, $\Prob[a,b]=\sum_{0\leq l\leq n_{a}} \Prob[C_{a}\A{l}]\cdot \Prob[a-1,b-{l\choose 2}]$.
\end{enumerate}
\noindent
By the above dynamic programming, we can compute $\Prob[i-1,l]$ for $0\leq l\leq k-1$ in polynomial time. \\
By the definition of $Y=i$, it is no hard to see that we can rewrite $\Prob[Y=i]$ as follows:
$$
\Prob[Y=i]=\sum_{1\leq n_1\leq |\mu'_{i+1}|, 1\leq n_2\leq |\mu''_{i+1}|} \Prob[\mu'_{i+1}\A{n_1}]\cdot \Prob[\mu''_{i+1}\A{n_2}]\cdot
\left( \sum_{k-{n_1+n_2\choose 2}\leq l\leq k-1-{n_1\choose 2}-{n_2\choose 2}} \Prob[\Gamma'_i\A{\#l}]\right)
$$
Note that we can compute $\Prob[Y=i]$ in polynomial time. We need to describe our sampler conditioning on $Y=i$. We first sample the event $\mu'_{i+1}\A{n_1}\wedge \mu''_{i+1}\A{n_2}$ with probability $\Prob[\mu'_{i+1}\A{n_1}\wedge \mu''_{i+1}\A{n_2}\mid Y=i]$. Then conditioning on $k-{n_1+n_2\choose 2}\leq l\leq k-1-{n_1\choose 2}-{n_2\choose 2}$, we sample the total number of point collisions in $\Gamma'_i$. Then we sample the number of present points in each component in $\Gamma'_i$ using the dynamic programming. Finally, based on the number of present points in each component, we sample the present points by Lemma~\ref{lm:samTj}. \\
Using the Monte Carlo method, we only need to take $O(\frac{m}{\e^2}\ln m)$ independent samples for estimating $\Exp[\KCP\mid Y=i]$. Thus, we totally take $O(\frac{m^2}{\e^2}\ln m)$ independent samples.
\qed
\end{proof}

\begin{theorem}
\label{thm:kcp}
There is an FPRAS for estimating the expected distance between
the $k$th closest pair in the existential uncertainty model.
\end{theorem}

\subsection{Hardness for Closest Pair}
\label{app:npcp}

\begin{theorem}
\label{thm:cpsharpp}
Computing $\Pr[\CP\geq 1]$  is \#P-hard
to approximate within any factor in a metric space
in both the existential and locational uncertainty models.
\end{theorem}
\begin{proof}
First consider the existential uncertainty model.
Consider a metric graph $G$ with edge weights being either $0.9$ or $1.8$.
Each vertex in this graph exists with probability 1/2.
Let $G'$ be the unweighted graph with the same number of vertices.
$G'$ contains only those edges corresponding to edges with weight 0.9 in $G$.
It is not hard to see that
$$
\Pr[\CP\geq 1] = \#\text{independent sets of size at least two in }G'\cdot \frac{1}{2^n}.
$$
The right hand side is well known to be imapproximable for arbitrary graphs
\cite{sly2010computational}.

For the locational model,
let the instance be $G$ (with $m$ vertices $s_1,\ldots, s_m$) with $m$ additional vertices $t_1,\ldots, t_m$
which are far away from each other and any vertex in $G$.
Let the probability distribution of node $v_i$ be $\p_{v_i s_i}=1/2$, and $\p_{v_i t_i}=1/2$.
We can see that in this locational uncertainty model, the value
$\Pr[\CP\geq 1]$ is the same as
that in the corresponding existential model $G$.
\qed
\end{proof}

\begin{theorem} Computing $\Exp[\CP]$ exactly in both the existential and locational uncertainty models is \#P-hard in a metric space.
\end{theorem}
\begin{proof}
Consider a metric graph $G$ with edge weights being either 1 or 2.
Each vertex in this graph exists with probability 1/2.
Note that
$$
\Exp[\C] = \Pr[\C=1]+2\Pr[\C=2]
=(\Pr[\C\leq 1]-\Pr[\C=0])+2(1-\Pr[\C\leq 1])
$$
Computing $\Pr[\C=0]$ can be easily done in polynomial time.
Computing $\Pr[\C\leq 1]$ in such a graph is as hard as counting independent sets
in general graphs, hence is also \#P-hard (as in Theorem~\ref{thm:cpsharpp}). So,
computing $\Exp[\C]$ is \#P-hard as well.

For the locational model, let the instance be $G$ (with $m$ vertices $s_1,\ldots, s_m$) with $m$ additional vertices $t_1,\ldots, t_m$
which satisfies $\dist(s_i,t_j)=\dist(t_i,t_j)=5$ ($1\leq i,j\leq m, i\neq j$).
Let the probability distribution of node $v_i$ be $\p_{v_i s_i}=1/2$, and $\p_{v_i t_i}=1/2$.
It is not hard to see that in this locational uncertainty model, the value
$\Exp[\C]$ is linearly related to the value $\Exp[\C]$ in the existential model $G$. Therefore, computing $\Exp[\C]$ is also \#P-hard in the locational uncertainty model.
\qed
\end{proof}

\section{Another FPRAS for MST}
\label{app:mst}

W.l.o.g., we assume that for each point, there is only one node that may be realized
to this point.
Our algorithm is a slight generalization of
the one proposed in \cite{kamousi2011stochastic}.
Let $\Exp[i]$ be the expected $\MST$ length conditioned on the event that
all nodes $\{v_1,\ldots,v_n\}$ are realized to points in $\{s_i,\ldots, s_m\}$
(denote the event by $\iin(i,m)$).
Let $\Exp'[i]$ be the expected $\MST$ length conditioned on the event that
all nodes $\{v_1,\ldots,v_n\}$ are realized to $\{s_i,\ldots, s_m\}$
and at least one node is realized to $s_i$.
We use $s\realize s$ to denote the event that node $v$ is realized to point $s$.
Note that
\begin{align*}
\Exp[i]=\Exp'[i]\Prob[\exists v, v\realize s_i\mid \iin(i,m)]+
\Exp[i+1]\Prob[\not\exists v, v\realize s_i\mid \iin(i,m)]
\end{align*}

For a particular point $s_i$,
we reorder the points $\{s_i,\ldots, s_m\}$ as $\{s_i=r_i,\ldots, r_m\}$ in increasing order of distance from $s_i$.
Let $\Exp'[i, j]$ be the expected $\MST$ length for all nodes
conditioned on the event that
all nodes are realized to $\{r_i,\ldots, r_j\}$ (denoted as $\iin'(i,j)$) and
$\exists v, v\realize s_i$.
Let $\Exp''[i, j]$ be the expected $\MST$ length for all nodes
conditioned on the event $\iin'(i,j) \wedge
(\exists v, v\realize r_i)\wedge (\exists s', s'\realize r_j)$.
We can see that
\begin{align*}
\Exp'[i,j]=&\Exp''[i,j]\Prob[\exists v', v'\realize r_j\mid \iin'(i,j), \exists v, v\realize r_i] \\
&+ \Exp'[i,j-1]\Prob[\not\exists v, v\realize r_i\mid \iin'(i,j), \exists v, v\realize r_i]
\end{align*}
It is not difficult to see the probability
$\Prob[\exists v', v'\realize r_j\mid \iin'(i,j), \exists v, v\realize r_i]$
can be computed in polynomial time.
Here we use the assumption that for each point, only one node that may realize to it.
Moreover, we can also take samples conditioning on event
$\iin'(i,j) \wedge (\exists v, v\realize r_i)\wedge (\exists v', v'\realize r_j)$.
Therefore $\Exp''[i,j]$ can be approximated within a factor of $(1\pm \epsilon)$ using
the Monte Carlo method in polynomial time since it is poly-bounded.
The number of samples needed can be bounded by $O\bigl(\frac{nm^2}{\e^2}\ln m\bigr)$.

We can easily generalize the above algorithm to the case where
$\sum_{j=1}^m p_{ij}\leq 1$, i.e., node $i$ may not be present with some certainty.
Indeed, this can be done by generalizing the definition of
$\iin(i,j)$ (and similarly $\iin'(i,j)$)
to be the event that each node is either absent or realized to some point in $\{r_i,\ldots, r_j\}$.

\end{document}